\documentclass[10pt,aps,prl, 
onecolumn,
superscriptaddress, notitlepage]{revtex4-1}

\usepackage{amssymb,amsmath}

\usepackage{natbib}
\usepackage{graphicx}
\usepackage{amsthm}
\usepackage{cancel}
\usepackage{color}
\usepackage{bbold}
\usepackage{wrapfig}
\usepackage{booktabs}
\usepackage{float}
\usepackage{enumitem}
\usepackage{hyperref}
\usepackage{tikz}
\usepackage{comment}
\usepackage{longtable}
\usepackage{multirow}
\usepackage{hyperref}

\newtheorem{theorem}{Theorem}

\usepackage[normalem]{ulem}

\begin{document}

\title{Incentive-driven transition to high ride-sharing adoption}
	
\author{David-Maximilian Storch}
\affiliation{Chair for Network Dynamics, Institute for Theoretical Physics and Center for Advancing Electronics Dresden (cfaed), Technical University of Dresden, 01062 Dresden, Germany}
\author{Marc Timme}
\affiliation{Chair for Network Dynamics, Institute for Theoretical Physics and Center for Advancing Electronics Dresden (cfaed), Technical University of Dresden, 01062 Dresden, Germany}
\affiliation{Lakeside Labs, Lakeside B04b, 9020 Klagenfurt, Austria}
\author{Malte Schr\"oder}
\affiliation{Chair for Network Dynamics, Institute for Theoretical Physics and Center for Advancing Electronics Dresden (cfaed), Technical University of Dresden, 01062 Dresden, Germany}

\date{\today}

\begin{abstract}
    Ride-sharing -- the combination of multiple trips into one -- may substantially contribute towards sustainable urban mobility. It is most efficient at high demand locations with many similar trip requests. However, here we reveal that people's willingness to share rides does not follow this trend. Modeling the fundamental incentives underlying individual ride-sharing decisions, we find two opposing adoption regimes, one with constant and another one with decreasing adoption as demand increases. In the high demand limit, the transition between these regimes becomes discontinuous, switching abruptly from low to high ride-sharing adoption. Analyzing over 360 million ride requests in New York City and Chicago illustrates that both regimes coexist across the cities, consistent with our model predictions. These results suggest that even a moderate increase in the financial incentives may have a disproportionately large effect on the ride-sharing adoption of individual user groups.
    \end{abstract}

\maketitle

\section*{Introduction}

Sustainable mobility \cite{un2015_sustainbleDevelopmentGoals, Helbing2001,Song2010, Schlaepfer2020, Gonzalez2008, Banister2008} is essential for ensuring socially, economically, and environmentally viable urban life \cite{IPCC2014, ECWhitePaper2011}. Ride-sharing (sometimes also referred to as ride-pooling) constitutes a promising alternative to individual motorized transport by private cars or single-occupant taxi cabs, currently dominating urban mobility \cite{Santi2014}. In ride-sharing, one vehicle transports multiple passengers simultaneously by combining two or more trip requests with similar origin and destination. In contrast to analog on-street hailing of taxi rides, digital app-based ride-hailing services are especially suited to implement ride-sharing due to easy access to the information required to match different trips.

By combining different individual trips into a shared ride, ride-sharing increases the average utilization per vehicle, reduces the total number of vehicles required to serve the same demand \cite{Vazifeh2018} and thereby mitigates congestion and negative environmental impacts of urban mobility \cite{Merlin2019, Jenn2020}. Hence, encouraging ride-sharing for trips that would otherwise be conducted in a single-occupancy motorized vehicle is preferable from a systemic perspective \cite{Anair2020, Jenn2020, Lopez2014, Santi2014b}.

Previous research focused on developing algorithms to implement large-scale ride-sharing \cite{alonso2017demand} as well as the potential efficiency gains derived from aggregating rides \cite{Santi2014, Tachet2017, Molkenthin2019}. Recent analyses suggest that large-scale ride-sharing is most efficient in densely populated urban areas \cite{agatz2012optimization, Vazifeh2018, Tachet2017, Santi2014, Molkenthin2019} since matching individual rides into shared ones without large detours becomes easier with more users, increasing both the economic and environmental efficiency as well as the service quality of the ride-sharing service \cite{Tachet2017, Herminghaus2019, Molkenthin2019}. 
Yet, if and under which conditions people are actually willing to adopt ride-sharing remains elusive \cite{Sarriera2017, margolin1978incentives, Schwieterman2018, Pratt2019, LO2018, Ruijter2020,Morris2019,Lippke2020}. In particular, it is unclear how to encourage an ever growing number of ride-hailing users to choose shared rides over their current individual mobility options \cite{Erhardt2019, Creutzig2018, TLC2018}.

In this article, we disentangle the complex incentive structure that governs ride-hailing users' decisions to share their rides -- or not. In a game theoretic model of a one-to-many demand constellation we illustrate how the interactions between individual ride-hailing users give rise to two qualitatively different regimes of ride-sharing adoption: one low-sharing regime where the adoption decreases with increasing demand and one high-sharing regime where the population shares their rides independent of demand. Analyzing ride-sharing decisions from approximately 250 million ride-requests in New York City and 110 million in Chicago suggests that both adoption regimes coexist in these cities, consistent with our theoretical predictions. Our findings indicate that a small increase in financial incentives may disproportionately increase the adoption of ride-sharing for individual user groups from a low to a high-sharing regime.

\section*{Results}

\subsection*{Contrasting ride-sharing adoption}

Currently, only a small fraction of people adopts ride-sharing even in high-demand situations, despite all its positive aspects \cite{McKinsey2019}. For example, among more than 250 million ride-hailing requests served in New York City in 2019 less than 18\% were requests for shared transportation \cite{data_NYC}. Moreover, the city's ride-sharing activity varies strongly across different parts of the city, in particular at locations with a high number of ride-hailing requests (see Fig.~\ref{fig:FIG1}): For instance, in the East Village and Crown Heights North the fraction of shared ride requests is relatively high, while it is low at both John F. Kennedy and LaGuardia airports, locations that would intuitively be especially efficient for sharing rides. Several other location throughout New York City as well as Chicago exhibit similarly contrasting ride-sharing adoption (see Supplementary Notes 1 and 2 for details).
These findings hint at a complex interplay of urban environment, demand structure and socio-economic factors that govern the adoption of ride-sharing. To disentangle these complex interactions, we introduce and analyze a game theoretic model capturing essential features of ride-sharing incentives, disincentives as well as topological demand structure.

\begin{figure}
    \centering
       \includegraphics[]{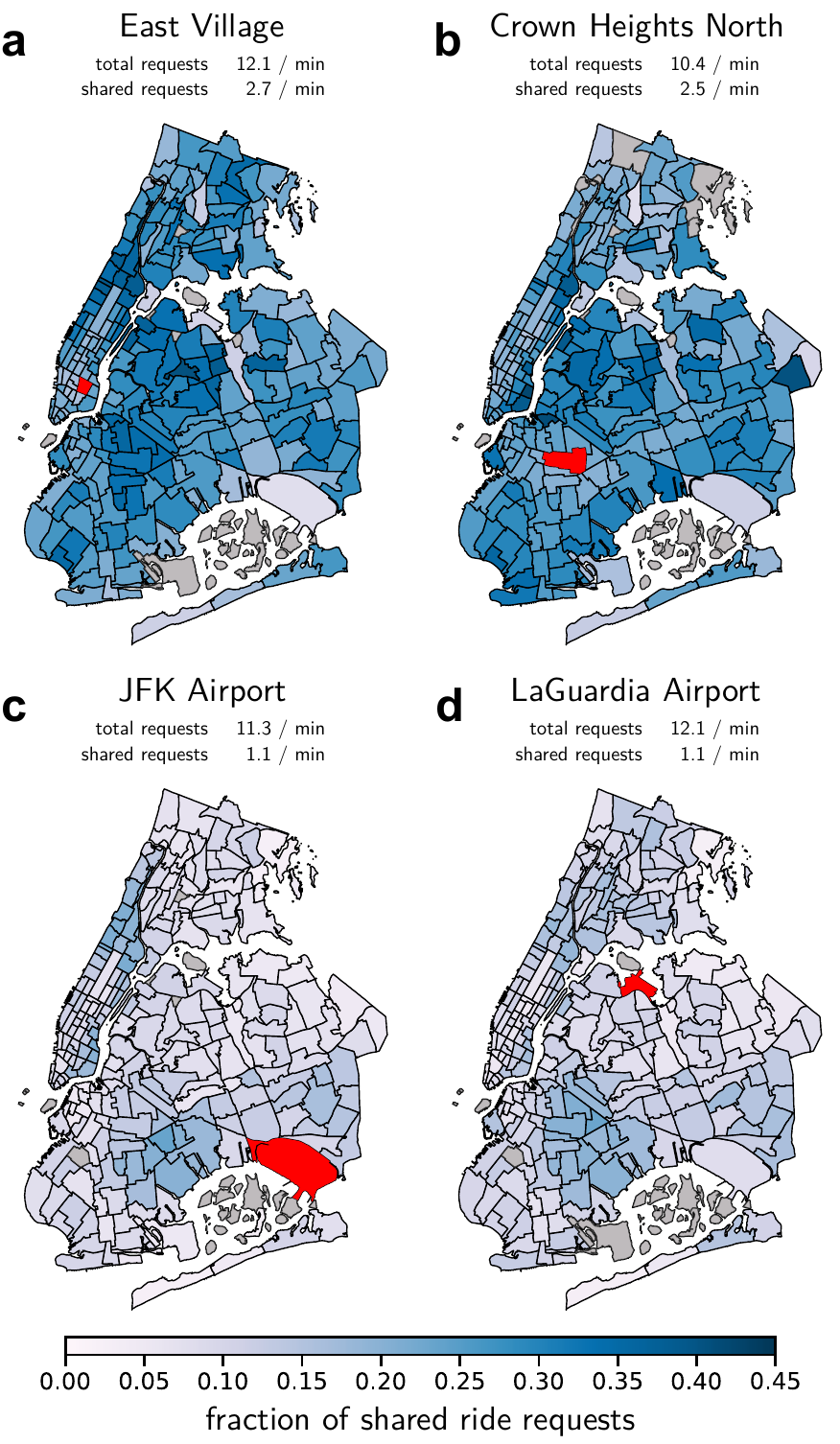}
    \caption{
\textbf{Contrasting ride-sharing adoption despite high request rate in New York City}.
Fraction of shared ride requests from different origins (red) served by the four major for-hire vehicle transportation service providers in New York City by destination zone (January - December 2019) \cite{data_NYC}. 
Gray areas were excluded from the analysis due to insufficient data (see Methods). The fraction of shared ride requests differs significantly by origin and destination with a complex spatial pattern across destinations, even though the average overall request rate is similar for all four origin locations. 
\textbf{a,b} Some areas, such as East Village and Crown Heights North, show a high adoption of ride-sharing services.
\textbf{c,d} Despite a similarly high request rate, other locations, such as JFK and LaGuardia airports, show a significantly lower adoption of ride-sharing services.
}
    \label{fig:FIG1}
\end{figure}

\subsection{Ride-sharing incentives}

\begin{figure}
    \centering
    \includegraphics{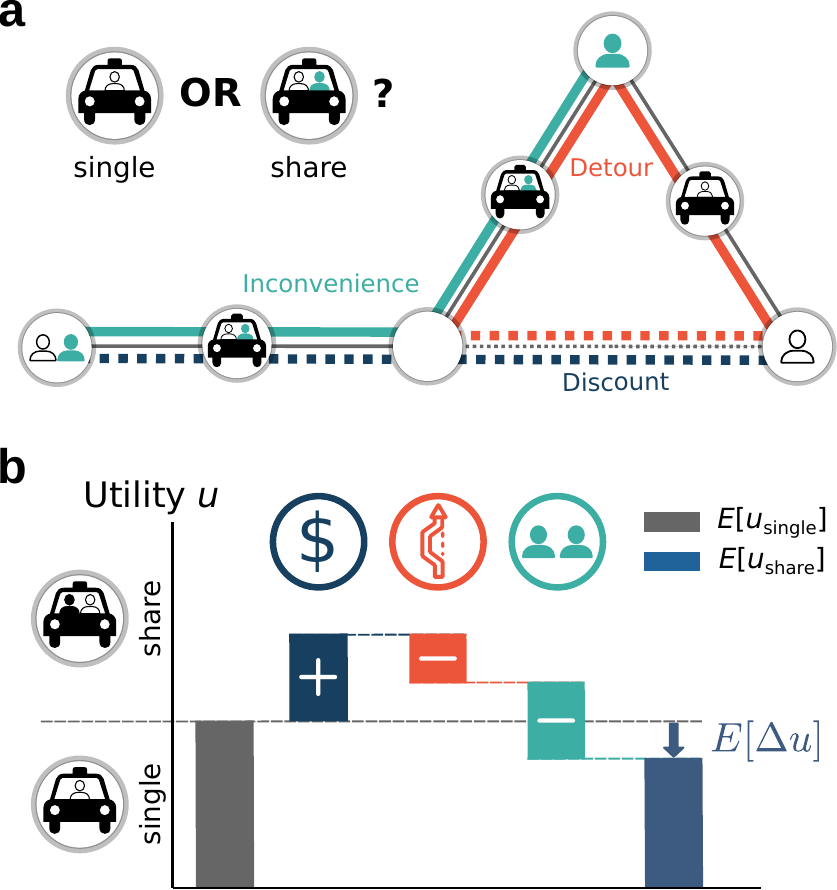}
    \caption{
    \textbf{Trade-offs between incentives determine the decision to share a ride, or not.} 
    \textbf{a} Shared rides offer advantages and disadvantages compared to single rides. 
    On the one hand, they offer financial discounts typically proportional to the distance of a direct single ride (blue, dotted).
    On the other hand, rides shared with strangers may include detours compared to a direct trip to pickup or deliver these other passengers (orange, solid compared to dotted) and may be inconvenient due to other passengers in the car (e.g. loss of privacy or less space,  green). 
    \textbf{b} The decision to book a shared ride depends on the balance of all three factors. If the expected utility difference $E[\Delta u] = E[u_\mathrm{share}] - E[u_\mathrm{single}]$ between a shared and a single ride is positive, the financial discounts overcompensate detour and inconvenience effects; users share. If $E[\Delta u]$ is negative (as illustrated), users prefer to book single rides.
    }
    \label{fig:FIG2}
\end{figure}

The decision of ride-hailing users to request a single or a shared ride reflects the balance of three fundamental incentives \cite{margolin1978incentives, Ruijter2020}: financial discounts, expected detours as well as uncertainty about the duration of the trip, and the inconvenience of sharing a vehicle with strangers. Strong correlations between the adoption of ride-sharing and (in)direct measures of the three incentives (see Supplementary Notes 1 and 2, including Supplementary Figures 5 and 7) confirm the importance of these incentives found in detailed empirical studies of ride-sharing user experiences as well as focus group interviews \cite{Morris2019, Sarriera2017, Lippke2020, LO2018, Schwieterman2018, Pratt2019}. Together, discounts, detours and inconvenience affect the ride-sharing adoption as follows (Fig.~\ref{fig:FIG2}):

\paragraph{Discounts.} Ride-sharing is incentivized by financial discounts granted on the single ride trip fare, partially passing on savings of the service cost to the user. Often, these discounts are offered as percentage discounts on the total fare such that the financial incentives $u_\mathrm{fin}^\mathrm{share} > 0$ are approximately proportional to the distance or duration $d_\mathrm{single}$ of the requested ride, $u^\mathrm{share}_\mathrm{fin} = \epsilon \, d_\mathrm{single}$, where $\epsilon$ denotes the per-distance financial incentives. In many cases, these discounts are also granted if the user cannot actually be matched with another user into a shared ride \cite{Uber_Pool, FreeNow_Match}.

\paragraph{Detours.} Potential detours $d_\mathrm{det}$ to pickup or to deliver other users on the same shared ride 
discourage sharing. The magnitude of this disincentive $u_\mathrm{det}^\mathrm{share} < 0$ increases with the detour $d_\mathrm{det}$.

\paragraph{Inconvenience.} Sharing a ride with another user may be inconvenient due to spending time in a crowded vehicle or due to loss of privacy \cite{margolin1978incentives, Pratt2019,LO2018}. This disincentive $u_\mathrm{inc}^\mathrm{share} < 0$ scales with the distance or duration $d_\mathrm{inc}$ users ride together.\\

In the following we take $u_\mathrm{det}^\mathrm{share} \propto d_\mathrm{det}$ and $u_\mathrm{inc}^\mathrm{share} \propto d_\mathrm{inc}$, describing the first order approximation of these disincentives and matching the linear scaling of the financial incentives with the relevant distance or time. 

These incentives for a shared ride describe the difference $\Delta u$ in utility compared to a single ride or another mode of transport. The overall utility of a shared ride is then given by
\begin{eqnarray}
    u_\mathrm{share} &=&  u_\mathrm{single} + \Delta u \nonumber \\
    &=& u_\mathrm{single} + u_\mathrm{fin}^\mathrm{share} + u_\mathrm{det}^\mathrm{share}  + u_\mathrm{inc}^\mathrm{share} \label{eq:utility} \\
    &=& u_\mathrm{single} + \epsilon \, d_\mathrm{single} - \xi \, d_\mathrm{det} - \zeta \, d_\mathrm{inc} \nonumber 
\end{eqnarray}
where the utility $u_\mathrm{single}$ for a single ride describes the benefit of being transported, as well as the cost and time spent on the ride. 
The factors $\epsilon$, $\xi$ and $\zeta$ denote the user's preferences. By rescaling the utilities (measuring in monetary units), $\epsilon$ directly denotes the relative price difference between single and shared rides whereas $\zeta$ and $\xi$ quantify the importance of inconvenience and detours 
relative to the financial incentives (see Supplementary Note 3 for details).

For a given origin-destination pair with fixed single ride distance $d_\mathrm{single}$, financial incentives are constant for a given discount factor $\epsilon$. In contrast, detour and inconvenience contributions depend on the destinations and sharing decisions of other users. Their magnitude depends on where these users are going and on the route the vehicle is taking for a shared ride (see Methods). The decision to share a ride is determined by the expected utility difference (see Fig.~\ref{fig:FIG2})
\begin{equation}
    E[\Delta u] = E[u_\mathrm{share}] - E[u_\mathrm{single}] \label{eq:utility_diff}
\end{equation}
where $E[\cdot]$ signifies the expectation value over realizations of other users' destinations and sharing decisions conditional on one's own sharing decision.

\subsection{Ride-sharing coordination game on networks}

\begin{figure}
    \centering
    \includegraphics[width=\linewidth]{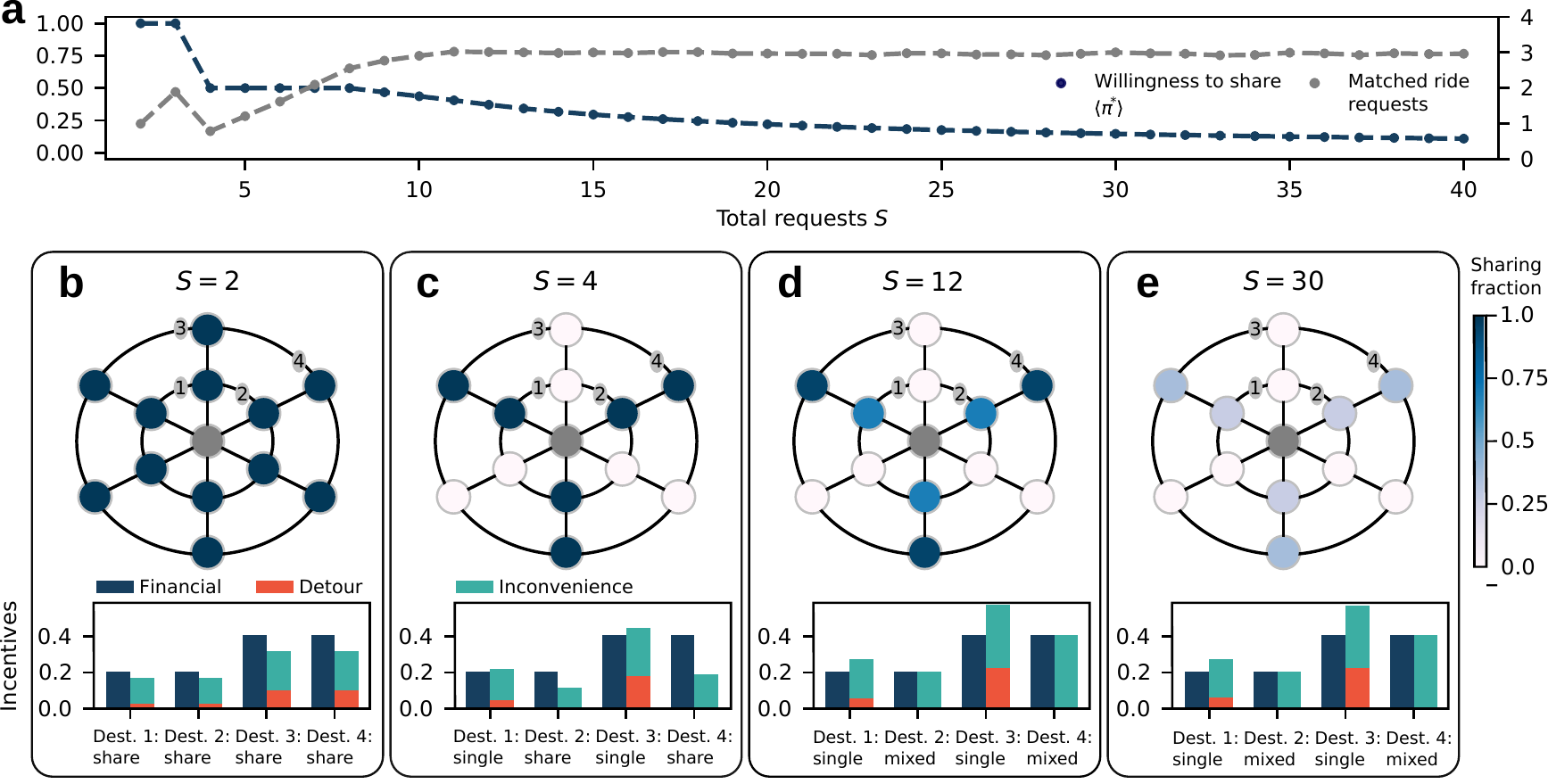}
        \caption{
    \textbf{Adoption of ride-sharing decreases with request rate.} In a stylized city topology (panels \textbf{b}-\textbf{e}) users request transportation from a single origin (gray) to destinations in the city periphery homogeneously (results are robust for alternative settings, see Supplementary Note 4). \textbf{a} The global equilibrium adoption of ride-sharing decreases as the number of users increases (blue) while the number of actually shared rides becomes constant (gray). 
    The kink for $S=3$ is an artefact related to the small and odd number of requests and matching of exactly two requests per vehicle such that one request can never be paired (see Supplementary Note 3 for details). \textbf{b-e} 
    As the number of users increases, ride-sharing adoption decreases and a sharing/non-sharing pattern emerges around the origin (top), resulting from the equilibrium incentive balance (bottom, illustrated for the numbered destinations) and possible matching constellations. Requests for shared rides are only matched when travelling to the same or to neighboring branches when the combined trip and return is shorter than the sum of individual trips. With few requests ($S=2$, panel \textbf{b}), all users request a shared ride. The expected detour and inconvenience is small since it is unlikely to be matched with another user. As the number of users increases ($S=4$, panel \textbf{c}), half of the destinations stop sharing in an alternating sharing/non-sharing pattern around the origin. 
    In this configuration, users requesting a shared ride never suffer any detour while users that do not share are disincentivized from doing so due to their high expected detour (compare bottom part of panel \textbf{c}).
    For high numbers of users ($S=12$ and $30$, panels \textbf{d} and \textbf{e}), the probability to be matched with another user when requesting a shared ride increases and the financial incentives cannot fully compensate the expected inconvenience. The adoption of ride-sharing decreases until the financial incentives exactly balance the expected inconvenience (panels \textbf{d} and \textbf{e}, bottom). 
    Illustrated here for financial discount $\epsilon = 0.2$ and inconvenience and detour preferences $\zeta = 0.3$ and $\xi = 0.3$.
}
    \label{fig:FIG3}
\end{figure}

To understand how these incentives determine the adoption of ride-sharing, we study sharing decisions in a stylized city network \cite{Lion2017} with a common origin $o$ in the center (e.g. a central downtown location) and multiple destinations $d$ (illustrated in Fig.~\ref{fig:FIG3}). Two rings define urban peripheries equidistant from the city center. Branches represent cardinal directions of destinations. Requests for shared rides will only be matched along adjacent branches, if the shared ride reduces the total distance driven to deliver the users and to return to the origin compared to single rides, consistent with a profit-maximizing service provider. Pairing at most two users who request a shared ride, the problem of matching shared ride requests reduces to a minimum-weight-matching with an efficient solution, eliminating the influence of heuristic matching algorithms \cite{Molkenthin2019, alonso2017demand} (see Methods for details).

In this one-to-many setting, users requesting a shared ride would only share a ride if they make their requests within some small time window $\tau$. Therefore, we consider a game with $S = s\,  \tau$ users travelling to a uniformly chosen destination location, where $s$ denotes the average request rate. These users have the option to book a single ride or a shared ride at discounted trip fare. Their decision to share depends on their expected utility difference $E[\Delta u(d)]$ [Eq.~\eqref{eq:utility_diff}], now depending on their respective destination $d$. Users observe their respective utility differences $E[\Delta u(d)]$ over a number of rides and adapt their sharing decision to maximize their expected utility. Eventually, users' sharing decisions converge to the equilibrium probabilities $\pi^*(d)$, reflecting an optimal response that maximizes the utility of users going to destination $d$ (see Methods for details).

At fixed discount $\epsilon$ and preferences $\zeta$ and $\xi$ ride-hailing users may decrease their overall adoption of ride-sharing $\langle \pi^* \rangle$ as the total number $S$ of users increases (see Fig.~\ref{fig:FIG3}a, blue), even though ride-sharing becomes more efficient with higher user numbers. Here $\langle \cdot \rangle$ denotes the average over all destinations $d$. While for small request rates everybody is requesting shared rides (Fig.~\ref{fig:FIG3}b), a distinctive sharing/non-sharing pattern emerges along the branches of the city network upon higher demand (Fig.~\ref{fig:FIG3}c,d), before the adoption of ride-sharing eventually fades out for high request rates, $S \gg 1$ (Fig.~\ref{fig:FIG3}e). This observation offers a novel perspective on the prevalent conclusion that increased demand improves the shareability of rides \cite{Santi2014,Molkenthin2019}. While more rides are potentially shareable, less people may be willing to share them.

The underlying incentives explain this phenomenon: 
Ideally, a user wants to book a shared ride (financial incentive) but without actually sharing the ride (inconvenience and detour). This discrepancy is consistently observable also in public vocalization of user sentiment about shared ride experiences \cite{Sarriera2017, Lippke2020, Pratt2019}, and exemplarily summarized by the user quote 'Every time I take a [shared ride] and it ends up being just me the entire ride I feel like a genius' \cite[pp. 112]{Morris2019}.
The expected detour and inconvenience mediate a repulsive interaction between the sharing decisions of ride-hailing users, turning ride-sharing decisions into a complex anti-coordination game. 
For small request rates, i.e. small numbers of concurrent users $S$, the probability $p_\mathrm{match}(d)$ for a user with destination $d$ to be matched with other users is low (see Fig.~\ref{fig:FIG3}a, gray). Consequently, the expected detour $E[d_\mathrm{det}(d)] = p_\mathrm{match}(d) \, E[d_\mathrm{det}(d) \,|\, \mathrm{match}]$ is also small (analogously for the inconvenience). As illustrated in Fig.~\ref{fig:FIG3}b, bottom, financial incentives outweigh the expected disadvantages of ride-sharing such that everybody is requesting shared rides, $\pi^*(d) = 1$ for all destinations $d$, but is only rarely matched with another user. 
As the number of users $S$ increases, the provider can pair ride requests more efficiently given constant sharing decisions, $\partial p_\mathrm{match}(d) / \partial S > 0$, resulting in more requests that are actually matched with another user (see Fig.~\ref{fig:FIG3}a). Consequently, the expected detour and inconvenience also increase. However, instead of reducing the average adoption of ride-sharing homogeneously across all destinations, neighboring destinations adopt opposing sharing strategies (see Fig.~\ref{fig:FIG3}b). In this sharing pattern, only destinations in identical cardinal direction can and will be matched into a shared ride, minimizing the detours for shared requests and simultaneously disincentivizing other users to start sharing due to high expected detours (Fig.~\ref{fig:FIG3}c-e bottom). As the number of users $S$ increases further, the probability $p_\mathrm{match}(d)$ would also increase at given sharing adoption $\pi(d)$. This leads to an adoption of mixed sharing strategies where the financial discounts and the expected inconvenience are exactly in balance (Fig.~\ref{fig:FIG3}d and e). 
Further numerical simulations demonstrate that this transition robustly exists also for heterogeneous demand distribution across the destinations, asymmetric street network topologies modeled by different origin locations within the network, for stochastic utility functions and imperfect information, as well as under different matching strategies by the service provider (see Supplementary Note 4 with Supplementary Figures 10-14, 17).

Naturally, if the discount $\epsilon$ is sufficiently large such that the financial incentives completely compensate the expected inconvenience, $\epsilon > \zeta$, all users share also in the high request rate limit, $S \rightarrow \infty$. In this limit, $d_\mathrm{single} = d_\mathrm{inc}$ as detours disappear, $E[d_\mathrm{det}] \rightarrow 0$, due to an abundance of similar requests. This transition is robust to changes of the model details, though under different matching strategies where detours remain possible in the high demand limit (see Supplementary Figure 17), the financial incentives required to achieve high sharing adoption may be larger.

Figure \ref{fig:FIG4}a summarizes these results in a phase diagram for the ride-sharing decisions as a function of financial discounts per inconvenience tolerance, $\epsilon/\zeta$, and number of users $S$, illustrating under which conditions the users adopt ride-sharing (high-sharing regime) and under which conditions the users only share partially or not at all (low-sharing regime). 

For fixed values of financial discounts $\epsilon$ relative to the inconvenience preference $\zeta$ of the users, different behavior emerges (Fig.~\ref{fig:FIG4}a): If $\epsilon/\zeta$ is sufficiently large ($\epsilon/\zeta > 1$), the system is in the high-sharing state and all users request a shared ride ($S_\mathrm{Share} = S$). Otherwise ($\epsilon/\zeta < 1$), the system transitions from the high- to a partial and finally to the low-sharing state (compare Fig.~\ref{fig:FIG3}). Figure~\ref{fig:FIG4}b illustrates the scaling of $S_\mathrm{Share}$ in both cases as $S$ increases. In the low-sharing regime, $S_\mathrm{Share}$ eventually becomes constant for large $S$, such that $S_\mathrm{Share} / S \rightarrow 0$ as $S \rightarrow \infty$ (compare Fig.~\ref{fig:FIG3}a). This implies a discontinuous phase transition between low-sharing and high-sharing regimes for large $S$ when the financial incentives exactly balance the inconvenience, $\epsilon_c/\zeta_c = 1$ (see Supplementary Note 3 and Supplementary Figure 9).

\begin{figure}
    \centering
    \includegraphics{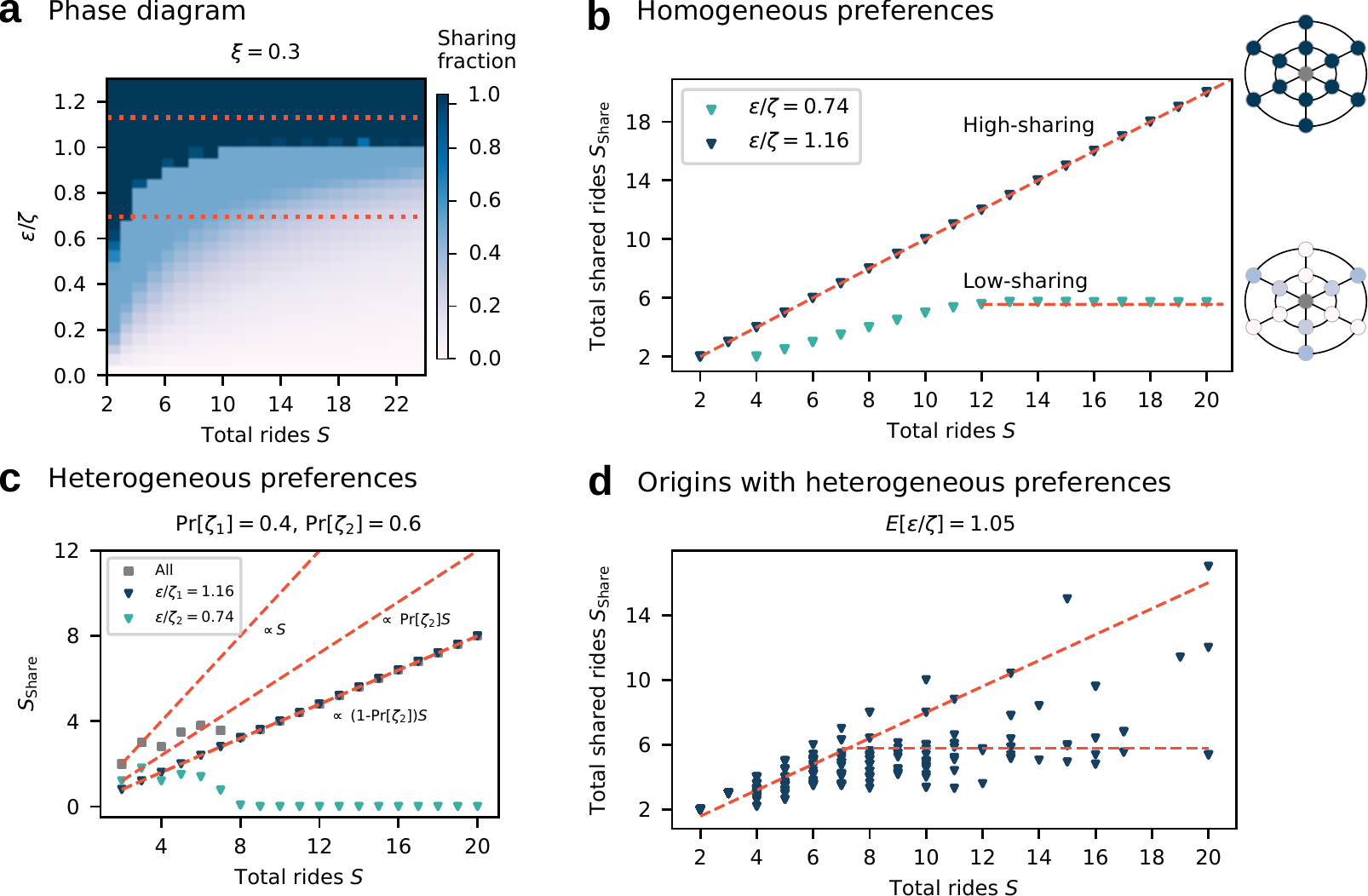}
    \caption{
    \textbf{Transition from low to high ride-sharing adoption.}
    \textbf{a} Phase diagram of the fraction of shared rides $S_\mathrm{share}/S$ for different relative importance of financial and inconvenience incentives $\epsilon/\zeta$. Ride-sharing is adopted dominantly if the financial discount fully compensates the expected inconvenience, $\epsilon/\zeta > 1$ (high-sharing, dark blue). 
    Otherwise, the total number of shared ride requests saturates and the overall adoption of ride-sharing decreases with increasing number of users $S$ (low-sharing, compare Fig.~\ref{fig:FIG3}a). In the limit of infinitely many requests $S \rightarrow \infty$ the transition becomes discontinuous (see Supplementary Note 3).
    \textbf{b} Qualitatively different sharing behavior emerges for different relative incentives $\epsilon/\zeta$ (compare red lines in panel \textbf{a}). When $\epsilon/\zeta > 1$ all users request shared rides ($S_\mathrm{share} = S$, dark blue triangles). When $\epsilon/\zeta < 1$, the system is in a low-sharing regime where users request shared rides at low numbers of users $S$ but the number of shared ride requests saturates and becomes constant at high $S$ ($S_\mathrm{share} < S$, light green triangles).
    \textbf{c} Hybrid states of high- and low-sharing adoption may emerge if users with heterogeneous preferences $\epsilon/\zeta$  mix and interact. A fraction of users (for whom $\epsilon/\zeta>1$) is in the high sharing regime (blue). The others (green, for whom $\epsilon/\zeta <1$) decrease their ride-sharing adoption as the overall demand increases, consistent with the prediction for homogeneous user preferences (panel b). Macroscopically, the system exhibits partial ride-sharing adoption (grey). 
    \textbf{d} The superposition of different realizations of this partial ride-sharing adoption represents the expected outcome in a city with multiple origins, each with heterogeneous preference distributions and demand (see Methods for parameters and Supplementary Note 4 for simulation details). While the macroscopic state suggests partial ride-sharing adoption, individual origins and user groups split into a mix of low- and high- sharing states, following the fundamental adoption regimes from the basic model.
    }
    \label{fig:FIG4}
\end{figure}

For heterogeneous preferences within the population (e.g. different preferences of the individual users requesting rides from the same location) the transition robustly persists per user type. If $\epsilon/\zeta_i < 1$ for parts of the local ride-hailing users, identified by their destination and preferences, these individuals transition from high- to low-sharing as the demand $S$ increases. The other part of the population, for whom $\epsilon/\zeta_i > 1$, remains in the high-sharing state. Macroscopically, the system appears to be in a partial-sharing state even at very high demand (compare Fig.~\ref{fig:FIG4}c), but in fact subsets of the population adopt opposing sharing strategies. The state of ride-sharing adoption across a city, i.e. across different origins each with a different distribution of inconvenience parameters and demand for rides $S$, is described by a superposition of these mixed states (see Fig.~\ref{fig:FIG4}d). Macroscopically, the system may appear to be in a hybrid state of partial- and low-sharing adoption, even when the aggregate population on average satisfies $E[\epsilon/\zeta]>1$ (see Methods, Supplementary Note 3 and Supplementary Figures 15-16 for simulation details).

\subsection{Ride-sharing activity in New York City and Chicago}

Ride-sharing adoption across different parts in New York City (taxi zones) and Chicago (community areas), illustrated in Fig.~\ref{fig:FIG5} (see Methods and Supplementary Notes 1 and 2 for details), matches the qualitative sharing behavior expected for multiple origins with heterogeneous preferences and demand (compare Fig.~\ref{fig:FIG4}d and Supplementary Note 4). 

At locations with a low request rate $s$, the number of shared ride requests increases approximately linearly with more requests, $s_\mathrm{Share} \sim s$. Though even in the low demand limit, the ride-sharing adoption in New York City and Chicago, corresponding to the diagonal branches in Fig.~\ref{fig:FIG5}a,b, is below 100\% (approximately $20\%$ in New York City and $40\%$ in Chicago). In terms of our ride-sharing game, the remaining fraction of requests for single rides may correspond to a user group with high relative importance of inconvenience compared to financial incentives, $\epsilon/\zeta \ll 1$, or that otherwise does not consider sharing as an option. In this interpretation, the smaller value for New York City is consistent with a large fraction of high-income and business customers in Manhattan who likely place a higher value on convenience than financial incentives.

\begin{figure}
    \centering
    \includegraphics{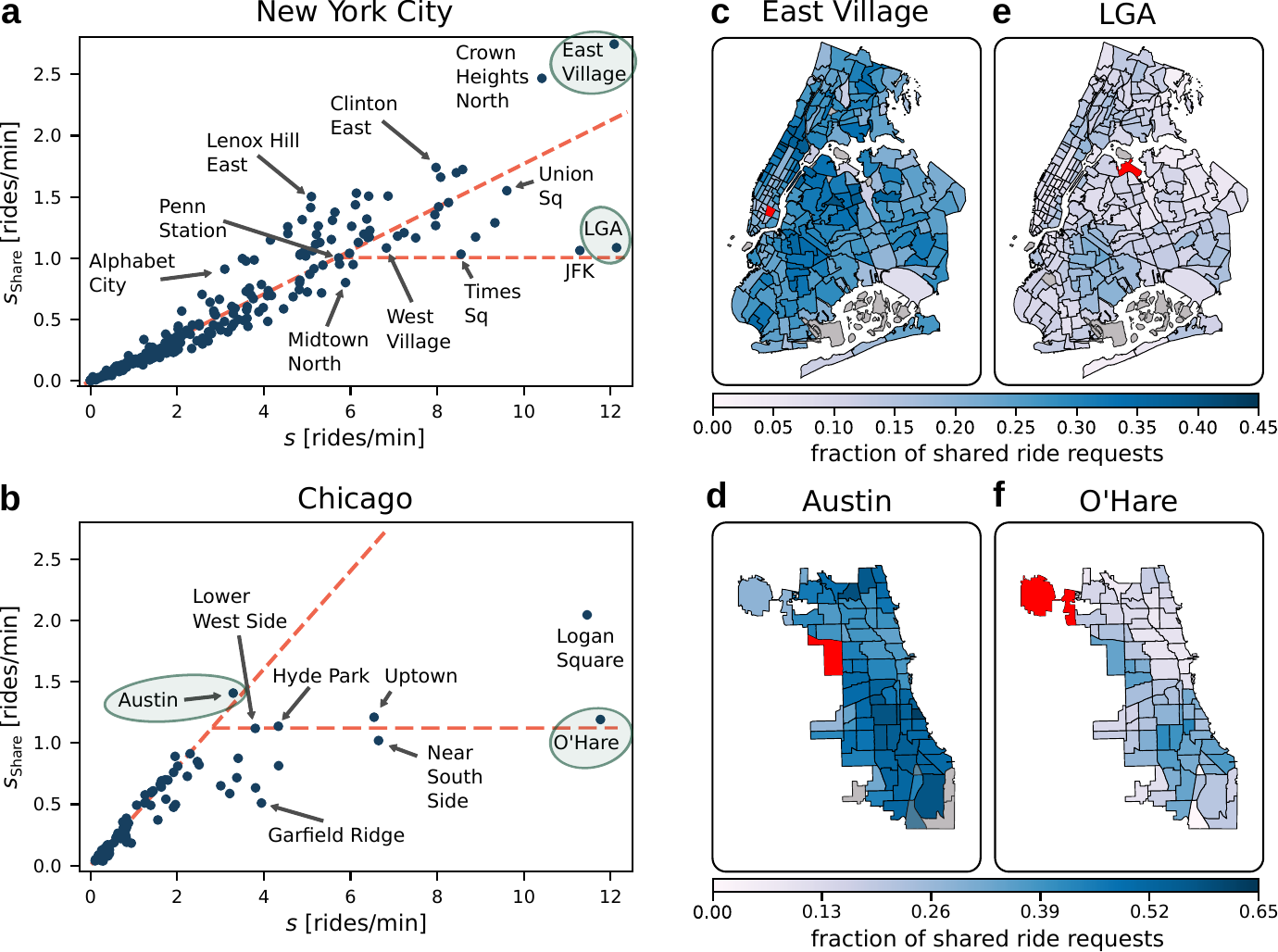}
    \caption{ \textbf{Ride-sharing adoption in New York City and Chicago is consistent with the predicted high- and low-sharing regimes.} 
    \textbf{a,b} 
    Sharing decisions for New York City and Chicago (blue dots) distribute between the two branches corresponding to the high- and low-sharing regime, consistent with the model predictions under heterogeneous user preferences (compare Fig.~\ref{fig:FIG4}).
    At low request rates, the number of requests for shared rides increases linearly with the total number of requests (compare red diagonal). At high request rates, the sharing decisions differ between locations (compare Fig.~\ref{fig:FIG1} and \ref{fig:FIG4}, see also Supplementary Note 1 and 2). As inconvenience preferences $\zeta$ are naturally heterogeneous in the cities, adoption is in a hybrid low/high-sharing state.
    \textbf{c-f} Ride-sharing adoption is consistently higher across destination zones in the high-sharing regime compared to the low-sharing regime. The predominantly linear increase of the number of shared rides in New York City as demand increases suggests broadly sufficient financial compensation of sharing disutilities, or, alternatively a very broad range of user preferences, leading to a stable fraction of ride-sharing adoption.
    However, the slope of the high-sharing branch indicates that only about $20\%$ of ride-hailing users consider ride-sharing as an option. 
    While about $40\%$ of requests are shared in the high-sharing regime in Chicago, this potential is largely not realized. 
    The available data points at locations with relatively high request rate indicate a growth with the request rate that is much weaker than on average for the entire data set, or even absent, consistent with the low-sharing regime observed in our model.
    Seven large downtown zones in Chicago with up to 50 requests per minute (not shown) fall in between the high- and low-sharing state, likely representing the average of sharing behavior over a diverse population of ride-hailing users as expected for users with heterogeneous preferences (see Supplementary Figure 6 for details).
    }
    \label{fig:FIG5}
\end{figure}

At higher request rates, sharing decisions differ by origin zone and split between low and partial sharing states (compare Fig.~\ref{fig:FIG1}). In New York City (Fig.~\ref{fig:FIG5}a), Crown Heights North and East Village exhibit a relatively high ride-sharing adoption in line with that observed in low demand zones, indicating $\epsilon$ is sufficiently large to compensate the expected inconvenience and detour effects for a significant fraction of the users. Other origins with a similarly high request rate, such as JFK and LaGuardia airports, do not follow this trend and exhibit a smaller number of shared ride requests. In terms of our model, we expect that $s_\mathrm{Share}$ has largely saturated in these zones and the given financial incentives do not outweigh the perceived inconvenience of ride-sharing. In particular at the airports, it seems plausible that financial incentives for ride-sharing are less important to users in the context of already costly plane tickets. 
In Chicago (Fig.~\ref{fig:FIG5}b), we find high-demand zones with an approximately constant number of shared ride requests, consistent with the low-sharing regime (horizontal branch $s_\mathrm{Share} = \mathrm{const.}$ in our model).
In contrast, no zones with high demand show the same, relatively high ride-sharing adoption as zones with low demand. Some large downtown zones in Chicago with up to 50 requests per minute fall in the partial sharing regime expected for zones that effectively aggregate sharing decisions over a broad distribution of user preferences.


\section*{Discussion}

Ride-sharing bears a large potential in the transition towards more sustainable mobility \cite{Santi2014, Tachet2017}. Yet, it remains poorly understood how to unlock this potential due to the complex interplay of demand patterns, matching algorithms, available transportation options, urban environments and the relevant incentive structure governing the adoption of shared rides. We have introduced a game theoretical model capturing incentives for and against ride-sharing from a user perspective, reflecting the major incentives found in empirical studies of users' ride-sharing experience \cite{Morris2019, Sarriera2017, Lippke2020, LO2018, Schwieterman2018, Pratt2019,Ruijter2020}. The model offers mechanistic insight into the collective effect of these incentives on individual ride-sharing decisions, unveils a discontinuous transition towards high overall ride-sharing adoption, and consistently explains the qualitative adoption of ride-sharing observed from 360 million empirical trip records from New York City and Chicago.

The ratio of financial discounts to inconvenience preferences acts as the control parameter in the model, separating two disparate regimes of ride-sharing adoption: one where the number of shared rides increases as the overall demand for rides increases (high-sharing regime) and one where it saturates (low-sharing regime), despite more efficient matching options and less detour as demand increases. Both regimes are separated by a regime with partial ride-sharing adoption that disappears in the high-demand limit. These results complement the finding of increased potential shareability of rides in high-demand settings \cite{Santi2014, Tachet2017} and may help to increase the service adoption to realize the full potential of ride-sharing under these conditions.

For homogeneous preference types across the user base, the adoption switches abruptly from the low adoption to the high adoption regime with a small change of the financial incentives and the transition between the two regimes becomes discontinuous (see Supplementary Note 3 for a mathematical proof). For heterogeneous preference types, as naturally expected in real cities with diverse population groups, the discontinuous transition robustly persists per user type. Macroscopically, however, heterogeneous preferences may induce a broad variance in the sharing adoption and yield mixed sharing decisions between the high- and -low sharing limit, blurring the abrupt transition towards high ride-sharing adoption as financial incentives increase. In line with our model predictions under spatially heterogeneous preferences, ride-sharing adoption observed in 360 million ride-sharing decisions from New York City and Chicago is broadly distributed across the cities, bounded between the high- and low-sharing regime (compare Fig.~\ref{fig:FIG4} and \ref{fig:FIG5}). Hence, the results above provide a consistent theoretical model and offer a possible explanation of qualitative features of ride-sharing adoption in urban environments, based on empirical model ingredients. The mechanisms captured by the model are independent of details of the incentive structure, utility functions, or matching and service scheme applied by the provider. We illustrate this robustness for a wide range of different conditions beyond those illustrated in Fig.~\ref{fig:FIG3}, including non-symmetric city topologies, heterogeneous demand distribution across possible destinations, noisy or imperfect information or decision-making, different strategies for matching rides, as well as for different simulation parameters (see Supplementary Note 4). However, deriving specific quantitative predictions from the model would require more detailed knowledge about users' preferences beyond the linear utility function assumed in our model. While the linear scaling in our model captures the basic features of the interactions, other models commonly assume a threshold dynamic to describe the impact of detour \cite{Santi2014, Tachet2017}. In addition, correlations in the demand structure and non-local matching of rides with different origins as well as the interplay between different service providers may also affect ride-sharing adoption. Similarly, the heterogeneity of ride-sharing adoption across different parts of the cities, expected in the low-sharing regime, seems to be dominated by socio-economic factors rather than by the pattern formation dynamics observed in our model network \cite{Wolf2021}.

Future research may investigate in more detail the impact of inconvenience on the adoption of ride-sharing, but also extend the analysis to additional factors such as users' sustainability attitudes, explicit risk aversion in the light of detour uncertainty, or mode choices with regard to public transportation alternatives. Our model description may already provide a theoretical framework for many of these factors influencing ride-sharing adoption on an aggregate level. For example, sustainability or uncertainty preferences to first approximation scale with the additional distance driven and may thus be effectively described by the detour preference. Similarly, alternative public transport options may be captured by modifying the effective financial discount and relative inconvenience preferences for individual destinations.

The sharp transition to high-sharing adoption predicted by our model for any given set of preferences of a user, suggests that even a moderate increase of financial incentives or a small improvement in service quality may disproportionately increase ride-sharing adoption of user groups currently in the low-sharing regime under a broad range of conditions. On the other hand, the overall low fraction of shared ride requests observed in the empirical trip records, even in the high-sharing regime, suggests that an additional societal change towards acceptance of shared mobility is required \cite{Sovacool2020} to make the full theoretical potential of ride-sharing accessible \cite{Santi2014, Tachet2017}. A carefully designed incentive structure for ride-sharing users adapted to local user preferences is essential to drive this change and to avoid curbing user adoption or stimulating unintended collective states \cite{Schroder2020, Helbing2000Freezing}. This is particularly relevant in the light of increasing demand as urbanization progresses \cite{un2015_sustainbleDevelopmentGoals}. In the broader context of macroscopic mode-choice behavior, e.g. between private car, ride-hailing or public transport, results and extensions of our model should be considered also from the perspective of rebound effects, such as more traffic induced by higher demand counteracting the benefits of ride-sharing. Nonetheless, the overall impact of more attractive ride-sharing on sustainability of urban transport is likely to be positive \cite{Morris2019, Jenn2020}. Overall, the approach introduced above can serve as a conceptual framework to work towards sustainable urban mobility by regulating and adapting incentives to promote ride-sharing in place of motorized individual transportation.

\clearpage


\section*{Methods}

\textbf{New York City ride-sharing data}. We analyzed trip data of more than 250 million transportation service requests delivered through high-volume For-Hire Vehicle (HVFHV) service providers in New York City in 2019. The data is provided by New York City's Taxi \& Limousine Commission (TLC) \cite{data_NYC} and consists of origin and destination zone per request, pickup and dropoff times, as well as a shared request tag, denoting a request for a single or shared ride. We compute the average request rate across all data throughout 2019 taking 16 hours of demand per day as an approximate average.

For fixed origin-destination pairs we determine the sharing fraction as the ratio of the total number of shared ride requests and the total number of requests. 
Departure and destination zones represent the geospatial taxi zones defined by TLC \cite{data_NYC}. However, we exclude zones without geographic decoding, nor name tag defined by TLC. For each individual analysis, we fix the origin zone and compute the fraction of shared rides to destination zones. 

To illustrate the spatial sharing adoption (shown in Fig.~\ref{fig:FIG1} and Fig.~\ref{fig:FIG5}c,e), we exclude destination zones where the total number of requests is less than 100 trips in the whole year 2019 to avoid excessive stochastic fluctuations (see Supplementary Note 1 and Supplementary Methods for details). We include these trips in the calculation of the average sharing fraction of the zone though they do not affect the averages due to their small number ($10^2$ compared to about $10^8$ trips in total).\\
 
\textbf{Chicago ride-sharing data}. We additionally analyzed more than 110 million trips delivered by three service providers in Chicago in 2019. The data is provided through the City of Chicago's Open Data Portal and contains, amongst others, information of trip origin, destination, pickup and dropoff times as well as information whether a shared ride has been authorized \cite{data_Chicago}. While information is available on whether a request was matched with another user, the flag denotes all consecutive trips where the vehicle was not empty, even if the passengers never shared part of their trip. We restrict ourselves to geospatial decoding of the city's 77 community areas, as well as trips leaving or entering the official city borders. In analogy to New York City, we compute the average request rate across all data for 2019 taking 16 hours of demand per day as an approximate average reference time and repeat the analysis explained for New York City.\\ 

\textbf{City topology}. For our ride-sharing model we construct a stylized city topology that combines star and ring topology \cite{Lion2017}. Starting from a central origin node, rides can be requested to $12$ destinations distributed equally across two rings of radius 1 (inner ring) and 2 (outer ring), as depicted in Figure~\ref{fig:FIG3}. The distances between neighboring nodes on the same branch are set to unity. Correspondingly, the distances between neighboring nodes are $\pi/3$ on the inner, and $2\pi/3$ on the outer ring.\\

\textbf{Ride-sharing adoption}. We compute the equilibrium state of ride-sharing adoption by evolving the adoption probabilities $\pi(d,t)$ following discrete-time replicator dynamics \cite{Cressman2014,Gaunersdorfer1995}
\begin{align}
    \pi(d,t+1) = r(d,t) \, \pi(d,t),
    \label{eqn:ReplicatorEqn}
\end{align}
where the reproduction rate $r(d,t)$ at destination $d$ and time $t$ is
\begin{align}
    r(d,t)&=\frac{E[u_\mathrm{share}(d,t)]}{E[u(d,t)]}= \frac{u_\mathrm{single}(d)+E[\Delta u(d,t)]}{u_\mathrm{single}(d)+\pi(d,t)E[\Delta u(d,t)]}
    \label{eqn:ReplicatorFactor}
\end{align}
and $E[X]$ represents the expectation value of random variable $X$. Conceptually, each user observes their utility difference between single and shared rides over a number of rides (e.g. using the service for week) and then adjusts their strategy $\pi(d,t)$ for the next time step. Users thus effectively learn their optimal equilibrium strategies where they cannot increase their utility by changing their decisions.

We realize this process in the following way:
We prepare the system in an initial state $\pi(d,0)=0.01$ of ride-sharing adoption for all destinations $d$, modeling the emergence of sharing. We fix the utility for a single ride $u_\mathrm{single}(d) = 4$ (unless stated otherwise) to ensure positivity of Eqn.~\eqref{eqn:ReplicatorFactor}. The value of $u_\mathrm{single}$ effectively controls the step size of the algorithm with $u_\mathrm{single} \rightarrow \infty$ corresponding to the continuous time limit of the replicator equation. The choice of $u_\mathrm{single}$ does not affect the equilibrium states ($\Delta u = 0$ or $\pi^* \in \{0,1\}$) and only determines the speed of convergence (compare Supplementary Fig.~18). 
To evolve Eqn.~\eqref{eqn:ReplicatorEqn}, we numerically compute $E[u_\mathrm{share}(d,t)]=E[u(d,t)|\mathrm{share}]$ at each replicator time step $t$: We generate $n=100$ samples of ride requests of size $S$ of which at least one goes to destination $d$ and requests a shared ride. The other $S-1$ requests
are drawn from a uniform destination distribution. Each of them realizes a sharing decision in line with the current probability distribution $\pi(d',t)$ at their respective destination $d'$ at time $t$. Shared ride requests are matched pairwise (see below). From these $n=100$ game realizations, we compute the conditional expected utility of sharing. We repeat this procedure for all destinations $d$ and then update all probabilities $\pi(d,t)$ according to Eqn.~\eqref{eqn:ReplicatorEqn}.

Before performing measurements on the system's equilibrium observables, we evolve the system for 20000 replicator time steps, corresponding to two million game realizations per destination. We discard a transient of 19000 replicator time steps and quantify the degree of fluctuations per $\pi(d)$ around its mean value over time for the last 1000 time steps. If fluctuations do not exceed a threshold of two percentage points we consider the system equilibrated. Else, we continue to evolve the system for another 5000 replicator time steps, test whether the equilibration threshold is met, and potentially repeat the procedure. The average ride-sharing adoption $\langle \pi(d) \rangle$ over the last 1000 replicator time steps represents a proxy for the stationary solution $\pi^*(d)$ of  Eqn.~\eqref{eqn:ReplicatorEqn} and is plotted as the sharing fraction in Figs.~\ref{fig:FIG3} and \ref{fig:FIG4}. In Supplementary Fig.~19 we quantify the degree of fluctuations per parameter constellation in the phase diagram in Fig.~\ref{fig:FIG4}a and demonstrate a high degree of equilibration, much better than the required threshold.\\

\textbf{Heterogeneous preferences}. Simulations for users with heterogeneous convenience preferences are carried out for fixed inconvenience parameters $\zeta_i$ for different user types. To determine the equilibrium ride-sharing adoption per user type, we repeat the equilibration procedure as explained in the previous paragraph, but the $S$ requests consist of randomly chosen user types with different inconvenience preferences. The probability to draw a user with preference $\zeta_i$ is given by the exogenous parameter Pr$[\zeta_i]$ (see Supplementary Note 4 and Supplementary Fig.~15). 

To produce Fig.~\ref{fig:FIG4}c we fix $\epsilon=0.2, \zeta_1=0.172$ and $\zeta_2=0.270$. The probabilities to draw $\zeta_1$ or $\zeta_2$ are Pr[$\zeta_1]=0.4$ and Pr[$\zeta_1]=0.6$, respectively. Other values yield qualitatively similar results (see Supplementary Note 4 for details). To compute the macroscopically observed combined contribution of shared ride requests from both user types (gray in Fig.~\ref{fig:FIG4}c), we sum the number of shared ride requests from the two user types for given total demand $S$.

To study the approximate macroscopic ride-sharing dynamics of a real city we superimpose 600 origin zones with different local demand for rides and local differences in convenience preferences of users (compare Fig.~\ref{fig:FIG4}d). Per origin we determine the local demand $S$ from an exponential distribution (see Supplementary Note 4 for details). Per origin, users may segment into three groups of convenience preference types $\zeta_i\in\{0.175,0.225,0.275\}$. The probabilities Pr$[\zeta_i]$ govern the distribution of convenience types per origin. Note that the distribution also determines for how many people, on average, $\epsilon$ overcompensates potential inconvenience effects.

Across origins we fix the macroscopic average ratio of financial incentives to inconvenience at $E[\epsilon/\zeta]=1.05$, hinting at a full-sharing state at the aggregated level. We draw the probabilities (Pr$[\zeta_1]$, Pr$[\zeta_2]$, Pr$[\zeta_3]$) from a normal distribution with mean $E[\epsilon/\zeta]=1.05]$ and standard deviation $\sigma=0.085$, fixing the local ratio of financial incentives to expected inconvenience parameters (see Supplementary Fig.~16).\\

\textbf{Matching}. Each request set of size $S$ decomposes into single and shared ride requests. We realize the optimal pairwise matching of requests as follows: For shared requests we construct a graph whose nodes correspond to requests and edges encode the distance savings potential of matching the two requests. To determine the distance savings potential we assume that, independent of single or shared ride, the provider has to return to the origin of the trip. 

After constructing the shared request graph we employ the 'Blossom V' implementation of Edmond's Blossom algorithm to determine the maximum weight matching of highest distance savings potential \cite{Kolmogorov2009}. The matching determines the routing and the realization of inconvenience and detour (see Supplementary Note 3 for more details). Since in the model all user requests are served, this matching strategy is consistent with a profit maximizing service provider.

\section*{Data Availability}

The trip record dataset for New York City is available in the Taxi \& Limousine Commission's (TLC) public repository \url{https://www1.nyc.gov/site/tlc/about/tlc-trip-record-data.page} as 'High Volume For-Hire Vehicle Trip Records' \cite{data_NYC}. The trip record dataset for Chicago is available on Chicago's Open Data portal \url{https://data.cityofchicago.org/Transportation/Transportation-Network-Providers-Trips/m6dm-c72p} as 'Transportation Network Providers - Trips' \cite{data_Chicago}. The simulation datasets generated from the game theoretical model in the current study are available upon reasonable request to the authors. 

\section*{Code Availability}
Full details on the data analysis and game theoretic modelling are provided in the Supplementary Information. The simulation code is available in the public Github repository 'PhysicsOfMobility/ridesharing-incentives', \url{https://doi.org/10.5281/zenodo.4630508} \cite{StorchGithub2021}.

\def\bibsection{\section*{\refname}}

\section*{Acknowledgements}
We thank the Network Science Group from the University of Cologne and Nora Molkenthin for helpful discussions and Christian Dethlefs for help with simulations. D.S. acknowledges support from the Studienstiftung des Deutschen Volkes. M.T. acknowledges support from the German Research Foundation (Deutsche Forschungsgemeinschaft, DFG) through the Center for Advancing Electronics Dresden (cfaed).

\section*{Author contribution} 

D.S. initiated the research with help from M.S. and M.T. All authors designed the research and provided methods and analysis tools. D.S. collected and analyzed the empirical data with help from M.S. D.S. and M.S. designed and analyzed the game theoretic model. All authors contributed to interpreting the results and wrote the manuscript.

\section*{Competing interest}
The authors declare no competing interests.

\clearpage

\setcounter{figure}{0}

\renewcommand{\figurename}{Supplementary Figure}
\renewcommand{\tablename}{Supplementary Table}
\def\bibsection{\section*{Supplementary References}} 

\section{Supplementary Note 1. Ride-sharing adoption in New York City}

In 2019, four high-volume for-hire vehicle (HVFHV) companies (Uber, Lyft, Via, Juno) served more than 250 million transportation service requests in New York City, corresponding to approximately 700,000 trips per day conducted by a population of 8.4 million people \cite{data_NYC,census_NYC}. In this Supplementary Note we explain the spatiotemporal demand distribution underlying this macroscopic number of transportation requests.

\subsection{Origin-destination demand}

The flux matrix $W(\Delta t)$ formalizes the spatiotemporal demand for transportation services between different locations of an urban environment. Its entries $W_{o,d}$ denote the number of transportation requests originating at location $o$ and going to location $d$ within a specific time window $\Delta t$.

For ride-hailing and ride-sharing services, $W(\Delta t)$ decomposes into
\begin{align}
    W(\Delta t) = W^\mathrm{single}(\Delta t)+W^\mathrm{shared}(\Delta t)
    \label{eqn:FluxMatrix}
\end{align}
where $W^\mathrm{single}(\Delta t)$ and $W^\mathrm{shared}(\Delta t)$ are the flux matrices describing trip requests tagged as single or shared rides, respectively. We define the fraction of rides shared as the relative ratio
\begin{align}
    P_{o,d}(\Delta t) = \frac{W_{o,d}^\mathrm{shared}(\Delta t)}{W_{o,d}^\mathrm{single}(\Delta t)+W_{o,d}^\mathrm{shared}(\Delta t)} = \frac{W_{o,d}^\mathrm{shared}(\Delta t)}{W_{o,d}(\Delta t)}, \qquad \mathrm{if}\ W_{o,d}(\Delta t) >w_\mathrm{min}
    \label{eqn:FractionSharedRides}
\end{align}
of shared rides to absolute number of rides between the same origin and destination. Note that Supplementary Eqn.~\eqref{eqn:FractionSharedRides} is only defined if the total flux between origin and destination exceeds a threshold $w_\mathrm{min}$ to reduce bias from fluctuations in statistical analyses of $P_{o,d}(\Delta t)$.

\subsection{Spatiotemporal ride-sharing activity in New York City}

\begin{figure}[!h]
    \centering
    \includegraphics[width=0.95\linewidth]{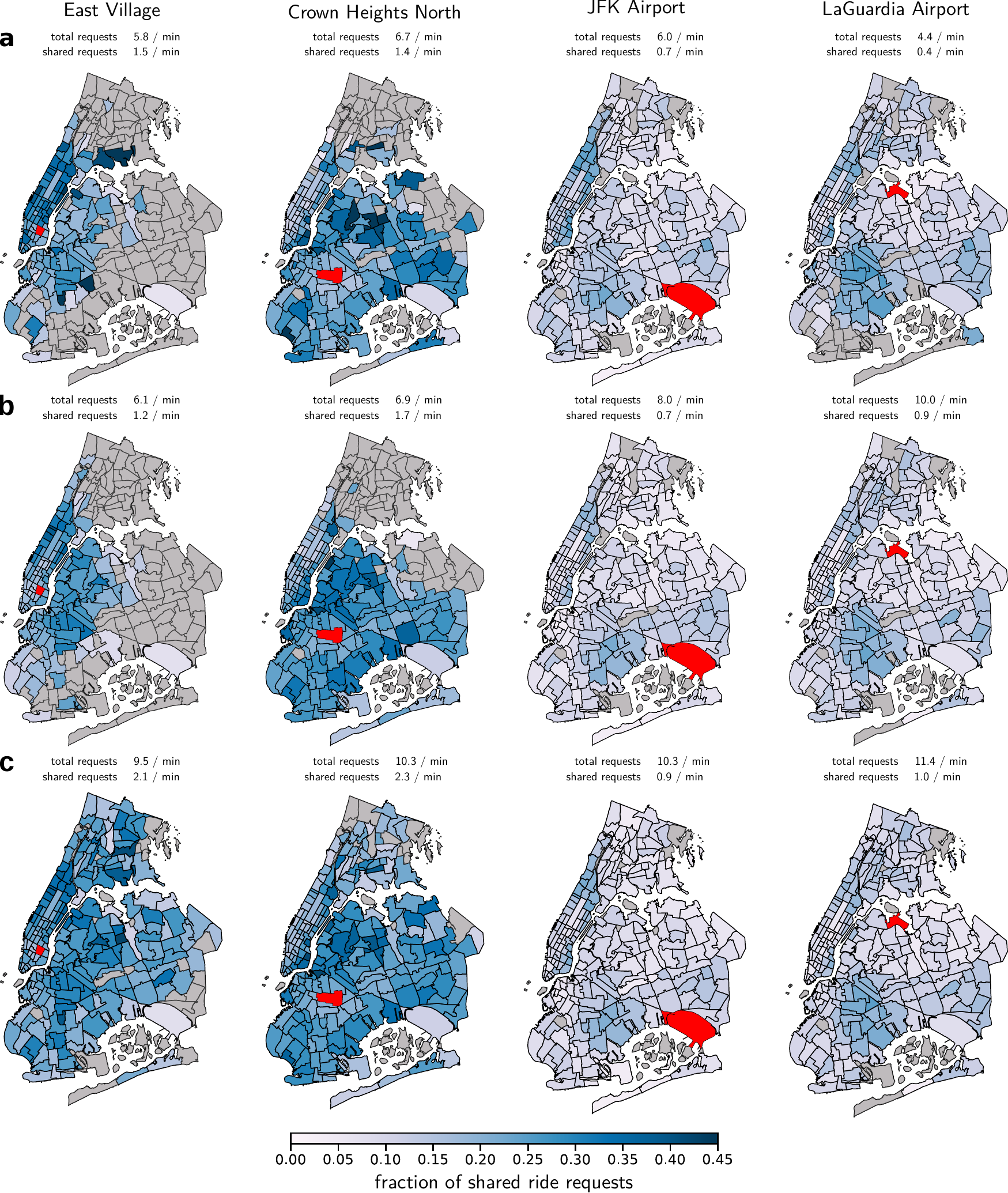}
    \caption{\textbf{Intraday variation of ride-sharing behavior in New York City.} 
    \textbf{a-c}, Morning (6-10 am), midday (12-4 pm) and evening (6 pm - 3 am) average adoption of ride-sharing for the four most requested origin zones of ride-hailing activity in New York City (red). Gray zones indicate insufficient total number of requests and are excluded from the analysis (see Supplementary Methods for details). 
    }
    \label{fig:SFIG1}
\end{figure}

We determine $W^\mathrm{single}, W^\mathrm{shared}$ and $P$ for the 265 taxi zones in New York City from the New York City Taxi \& Limousine Commission's (TLC) HVFHV aggregate trip records between January and December 2019 (compare Fig.~1 in the Main Manuscript, see Data section in Supplementary Methods for more details on data acquisition, structure and treatment). If not stated otherwise, we choose a value $w_\mathrm{min}=100$ of rides per year much smaller than the total number of rides (see Supplementary Methods for details).

Supplementary Figure \ref{fig:SFIG1} illustrates $P_{o,d}$ for the four origin zones $o$ with highest average demand for transportation services, $\sum_{\substack{d=1 \\ d\neq o}}^{265} W_{o,d}$, for three different time windows: morning (6-10 am) including commuting hours (Supplementary Fig.~\ref{fig:SFIG1}a), midday (12-4 pm) reflecting off-peak afternoon hours (Supplementary Fig.~\ref{fig:SFIG1}b) and evening (6 pm - 3 am) encompassing leisure activity hours (Supplementary Fig.~\ref{fig:SFIG1}c). Independent of daytime, all four origins exhibit complex spatial patterns of ride-sharing adoption across destinations. 

For JFK and LaGuardia airport these patterns are robust for all time windows, indicating stable fraction of rides shared to all destinations throughout the day. For Crown Heights North and East Village only few rides are undertaken to far distance destinations in the morning and midday time window (gray areas representing $W_{o,d}<w_\mathrm{min}$).
In the evening, more rides are requested overall, also to far distance destinations. Overall, the qualitative patterns of ride-sharing adoption do not vary significantly with the time of day (compare Fig.~1 in the Main Manuscript).

Across the full set of origin zones in New York City, Supplementary Figure~\ref{fig:SFIG2} suggests an overall trend to higher absolute demand for transportation services in the evening. The fraction of shared rides, however, is not affected by this trend. It is approximately constant throughout the day as illustrated in Supplementary Figure~\ref{fig:SFIG2b}. The average standard deviation of fraction of rides shared across all taxi zones is less than 1.9\% between the three time windows, suggesting an equilibrated system.

An aggregate analysis will naturally be dominated by the high overall demand in the evening and night time. Still, the data suggests that the average ride-sharing adoption in New York City is stable across the day. Hence, an aggregate analysis is representative.

\begin{figure}
    \centering
    \includegraphics[width = 0.95\linewidth]{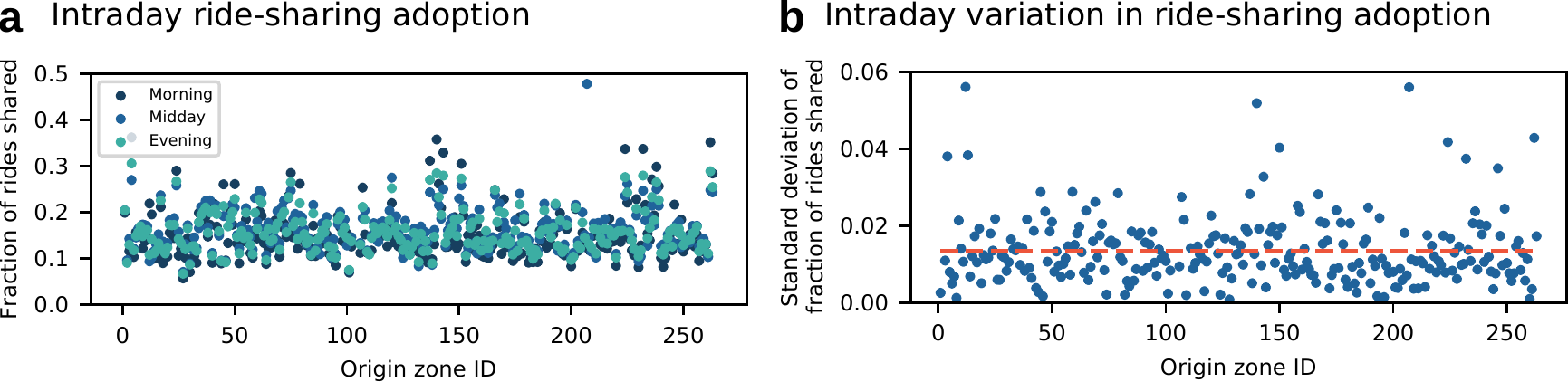}
    \caption{\textbf{Stable ride-sharing adoption intraday in New York City}. (left) During morning, midday and evening time windows the fraction of shared rides is approximately the same for the different origin zones. As before, origin-destination pairs with less than $w_{\min}$ total requests have been excluded from the analysis. (right) The standard deviation of the fraction of rides shared  across the three time intervals hints at average intraday fluctuations of less than 2\% in the fraction of rides shared (red dashed line).
    }
    \label{fig:SFIG2b}
\end{figure}

\subsection{Scaling properties of shared ride requests}

Supplementary Figure~\ref{fig:SFIG2} illustrates a dominant linear scaling in the total number of ride requests $S$ and the number of shared requests $S_\mathrm{Share}$ across origin zones in New York City for all time windows (morning, midday, evening). Such a linear scaling between $S_\mathrm{Share}$ and $S$ indicates sufficient financial incentives to compensate the expected negative effects of ride-sharing with increasing demand. A decrease in slope and eventual saturation corresponds to a situation where financial incentives, expected detour and inconvenience are in balance. $S_\mathrm{Share}$ will not increase upon higher demand for given incentives (compare Fig.~1 in Main Manuscript as well as large $S$ regime in Supplementary Fig.~\ref{fig:SFIG2}).

\begin{figure}[h]
    \centering
    \includegraphics[width=0.95\linewidth]{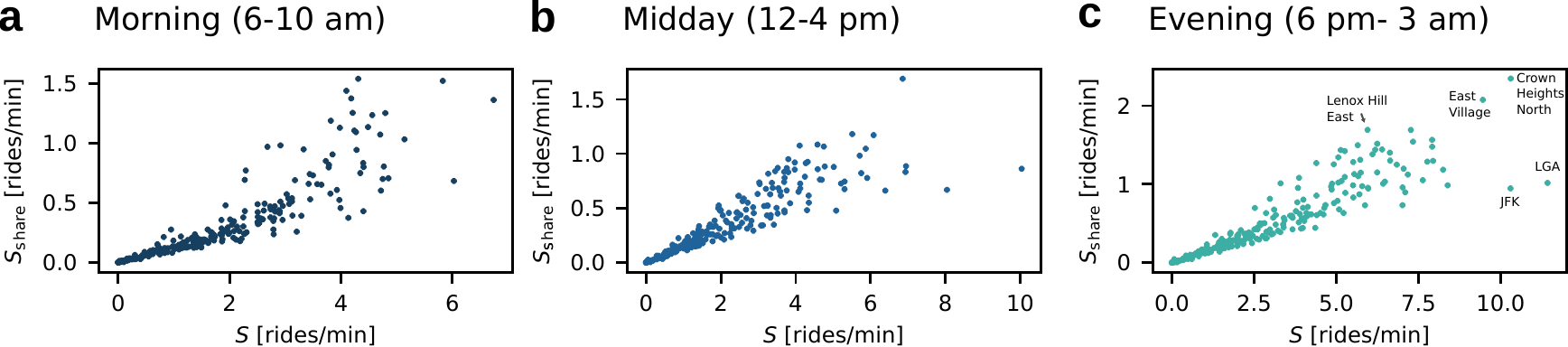}
    \caption{\textbf{Intraday sharing decisions in New York City reflect ride-sharing adoption close to full-sharing regime}. During morning, midday and evening time windows the number of shared rides scales approximately linear in the total number of ride requests, indicating sufficient financial incentives to compensate expected detour and inconvenience effects. For locations in the high request rate limit (compare Fig.~5 in the Main Manuscript) a saturation of $S_\mathrm{Share}$ indicates insufficient financial incentives for given negative aspects of sharing.
    }
    \label{fig:SFIG2}
\end{figure}

\clearpage

\subsection{Ride-sharing adoption at different origin zones in New York City}

Supplementary Figure~\ref{fig:SFIG3} shows daily averages for $W_{o,d}$ for the origins in New York City highlighted in Figure~5a of the Main Manuscript, assuming 16 hours of daily activity as a normalization for the time window $\Delta t$. Supplementary Figure~\ref{fig:SFIG3} (top row) represents origins with lower average ride-sharing adoption, for which the number of shared rides may have saturated (compare horizontal line in Fig.~5a in Main Manuscript). Supplementary Figure~\ref{fig:SFIG3} (bottom row) illustrates origins with relatively higher ride-sharing adoption on the linear ascending branch of the ride-sharing adoption curve in Fig.~5a in Main Manuscript.

\begin{figure}
    \centering
    \includegraphics[width=1.0\linewidth]{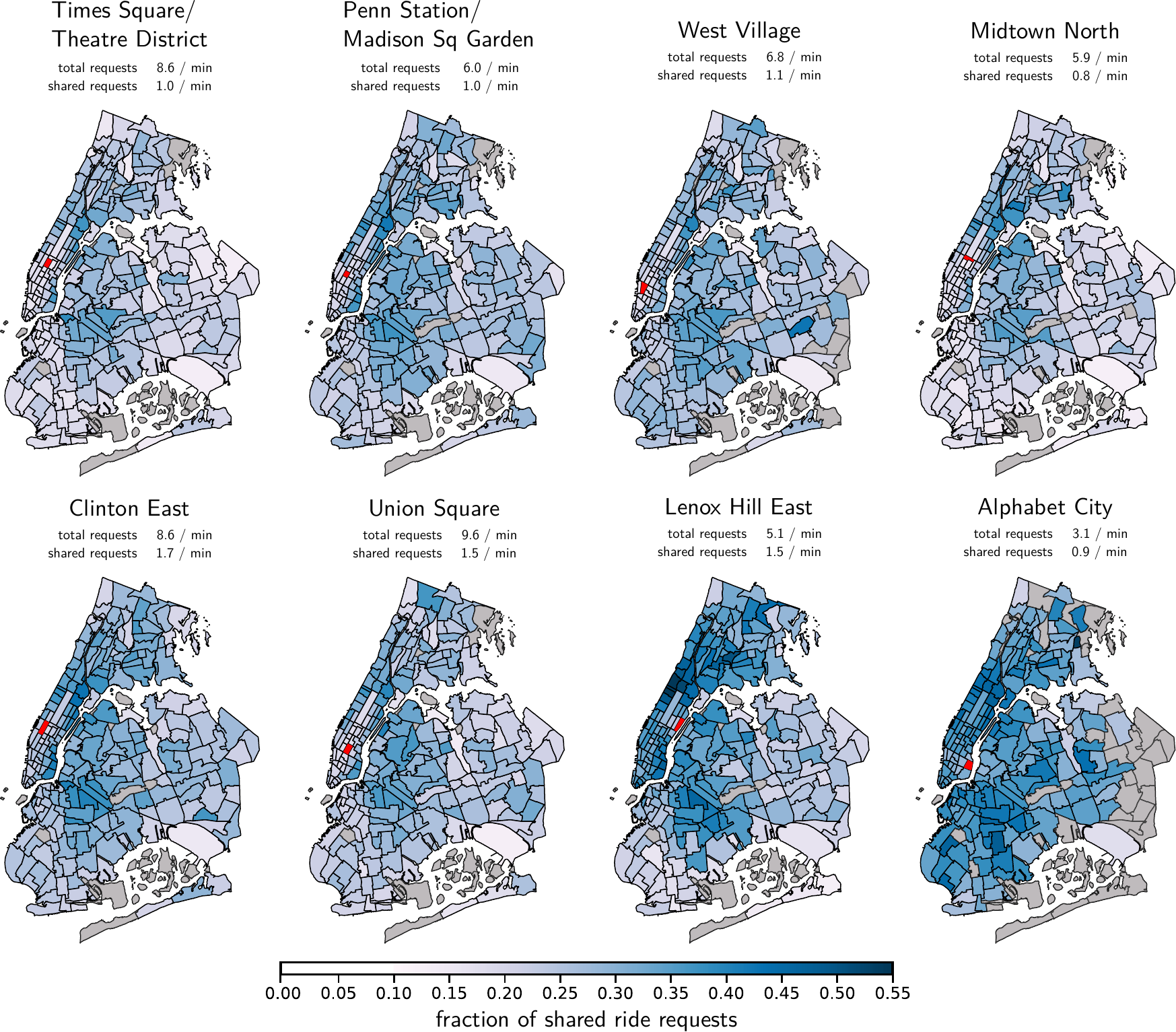}
    \caption{\textbf{Ride-sharing adoption is heterogeneous across New York City.} Origin zones with relatively low ride-sharing adoption (top row, compare horizontal line in Fig.~5a in Main Manuscript) yield similar absolute number of shared rides, despite differences in absolute request rate. Origin zones with yet unsaturated ride-sharing adoption (bottom row, compare Fig.~5a in Main Manuscript) yield larger average fraction of shared rides.}
    \label{fig:SFIG3}
\end{figure}

Consider for example Times Square/Theatre District and Alphabet City (top left and bottom right): While for the first only approximately one in nine ride requests is shared, it is one in three for the latter.

\clearpage

\subsection{Trip and ride-hailing user characteristics in New York City}

\begin{figure}
    \centering
    \includegraphics{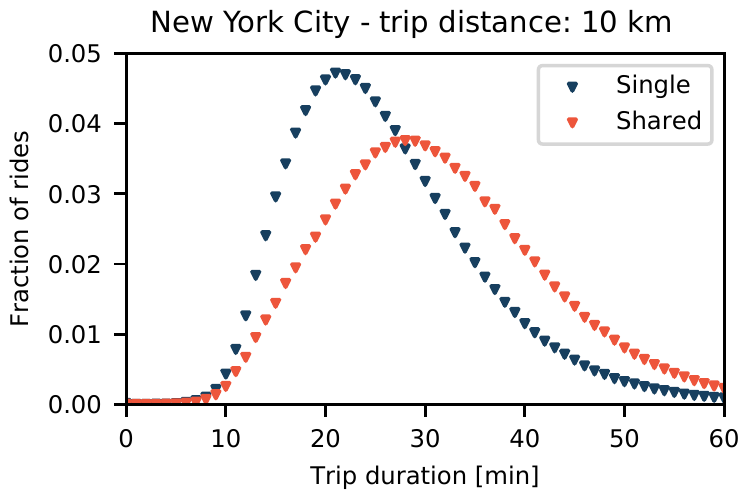}
    \includegraphics{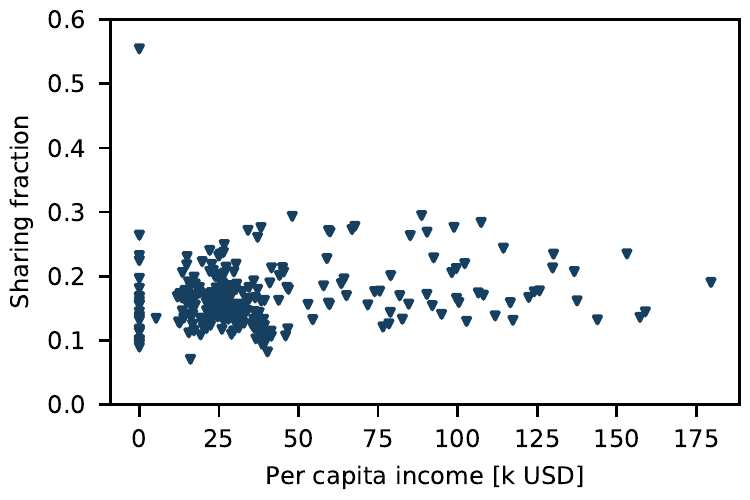}
    \caption{\textbf{In New York City shared rides likely take longer and are adopted largely independent of income.} (Left) Shared rides take longer on average and have higher trip duration variance compared to their single ride equivalents [mean trip duration for single rides between zones with distance between 9.75-10.25 km: 26.4 mins, for corresponding shared rides: 31.5 mins]. (Right) Ride-sharing adoption within a taxi zone and per capita income are only weakly correlated for high incomes exceeding fifty thousand USD [census data based on \cite{USCensus2020}]. Zero income zones include largely uninhabitated areas such as airports, Ellis Island, Jamaica Bay, public parks etc.
    }
    \label{fig:SFIG_NYC_TripChars}
\end{figure}

Empirical studies on user sentiment and focus group interviews with ride-sharing users have demonstrated that financial incentives, detours and uncertainty about the trip duration, as well as inconvenience from sharing a vehicle with strangers are dominant drivers of ride-sharing adoption \cite{Sarriera2017, Morris2019, Lippke2020}. In this section of the Supplementary Information we provide supporting evidence for these conclusions based on trip records from New York City (2019), as well as supporting census data on aggregated socio-economic characteristics of potential user groups.\\

\paragraph{Detour and trip duration uncertainty.} Supplementary Figure \ref{fig:SFIG_NYC_TripChars} (left) illustrates the trip duration distribution for single and shared rides of the same distance. In absence of precise information about trip distances in the available data records we estimate the trip distance as the shortest path distance between pickup and destination zone centroids. Single rides exhibit a trip duration distribution with positive skewness, while for shared ride the distribution is more symmetric. On average, single rides take less time compared to shared rides, likely due to detours from picking up or dropping off additional passengers in the case of a successful matching [approximately 5.1 minutes for a ten km ride in 2019]. For the passenger, the lost travel time may be associated with opportunity cost and thus be regarded as a source of disutility. Additionally, travel time uncertainty, estimated in terms of trip duration variance, is higher for shared rides. Hence, it is more difficult for ride-sharing users to anticipate actual arrival times at their final destination, leading to another form of disutility as frequently reported in empirical studies of ride-sharing user sentiment \cite{Morris2019}.

Other trip distances and other cities exhibit similar trip characteristics (see Supplementary Note 2 and Supplementary Figure \ref{fig:SFIG_CHI_TripChars} for the same analysis in Chicago).\\

\paragraph{Financial incentives.} In New York City the ride-sharing adoption is only weakly correlated with the per capita income if exceeding values of 50 thousand USD (see Supplementary Fig.~\ref{fig:SFIG_NYC_TripChars} (right)). Per capita income may be interpreted as a proxy for the relative importance of financial discounts in the decision to book a single or shared ride. Empirical studies and focus group interviews with ride-sharing users consistently report financial incentives to be a dominant driver of sharing adoption \cite{Morris2019, Lippke2020}. Hence, the nearly constant ride-sharing adoption across incomes in New York City may suggest limited leeway to increase the overall ride-sharing adoption through generally higher financial incentives, since financial stimulation may have already activated significant proportions of user potential considering ride-sharing a feasible transportation option.

\clearpage

\section{Supplementary Note 2. Ride-Sharing adoption in Chicago}

In 2019, three transportation service providers (Uber, Lyft, Via) served more than a total of 110 million transportation service requests in the City of Chicago, corresponding to approximately 300,000 trips per day \cite{data_Chicago}. In this Supplementary Note we demonstrate that the ride-sharing adoption in the city exhibits qualitative features of low- and high-sharing adoption states with complex spatial patterns.

Chicago consists of 77 community areas \cite{data_Chicago}. Supplementary Fig.~\ref{fig:SFIG4_Chicago} illustrates request rate for shared rides as a function of the total request rate for rides.
As illustrated for New York City in the Main Manuscript, Chicago's different communities exhibit spatially heterogeneous ride-sharing adoption. While there exists a subset of communities for which the number of shared ride requests scales linearly in the total number of requests, other origin communities (e.g. Lower West Side, Hyde Park, Uptown, Near South Side, O'Hare) form a branch where the number of shared ride requests has saturated and does not increase with the overall number of ride requests. Similarly to New York City, we observe partial and high-sharing regimes of ride-sharing adoption (compare Supplementary Fig.~\ref{fig:SFIG4_Chicago}b,c right).

The high demand zones shown in the inset of Supplementary Fig.~\ref{fig:SFIG4_Chicago} represent large zones from the city center with extremely high-demand of up to 50 requests per minute (North East Side, Loop, Near West Side, Lake View, West Town, Lincoln Park) with an intermediate sharing fraction. It is likely that Chicago's high-demand zones summarize the ride-sharing behavior of a wide range of the population with heterogeneous ride preferences (i.e. a mix between high and low sharing), naturally representing the average sharing fraction in the city. Comparable city districts in New York City (Manhattan) are split into more smaller zones, and therefore form into the heteroskedastic cone observed at small demand.

\begin{figure}
    \centering
    \includegraphics[width = 0.7 \linewidth]{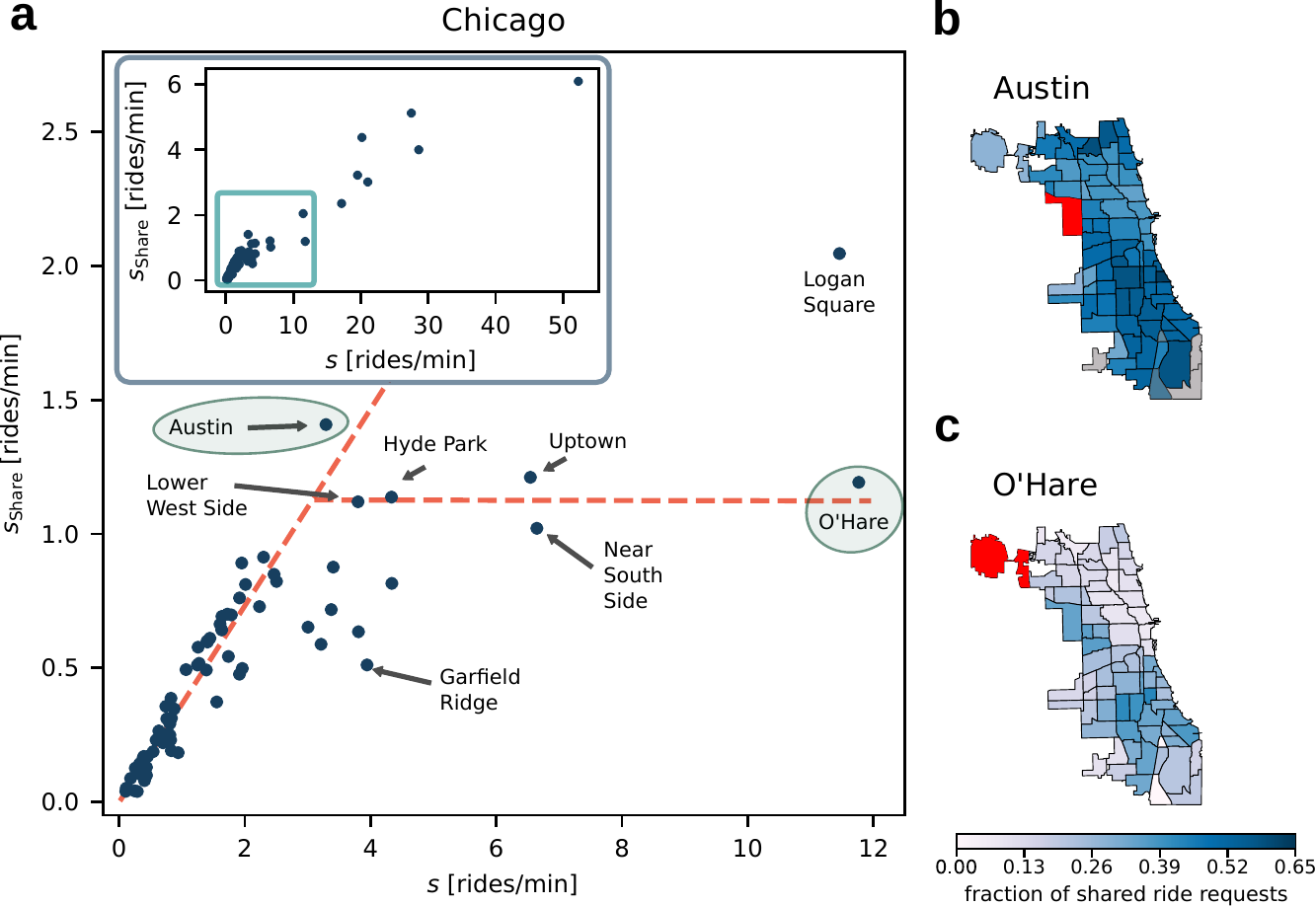}
    \caption{
    \textbf{Chicago exhibits hybrid sharing adoption}. \textbf{a} Sharing decisions in Chicago exhibit a hybrid state between low- and partial-sharing states (blue dots). At low request rates, the number of requests for shared rides increases linearly with the total number of requests. 
    At higher request rates, the sharing decisions do not increase as fast as the demand does, indicating a partial sharing regime. Few communities cross a saturation value (horizontal orange line), hinting at hardly any zones in a full-sharing regime, but generally low adoption of ride-sharing for the given financial incentives. The inset in panel a includes large community zones from the city center with extremely high-demand (North East Side, Loop, Near West Side, Lake View, West Town, Lincoln Park), and trips originating outside of the boundaries of the City of Chicago, whose request rates significantly exceed those of the other communities by up to one order of magnitude (not shown in the main panel, green border).  \textbf{b,c} Trips originating from Austin exhibit high average ride-sharing adoption of up to two thirds of all trips, in contrast to trips starting from O'Hare Airport which are much less frequently shared.
    }
    \label{fig:SFIG4_Chicago}
\end{figure}

\subsection{Trip and ride-hailing user characteristics in Chicago}

Trip records from Chicago as well as aggregated census data provide insight into important determinants of ride-sharing adoption. Similar to New York City, this section of the Supplementary Information provides supporting empirical evidence for detours and trip duration uncertainty, financial incentives and convenience being important aspects in ride-sharing decision-making.\\

\begin{figure}
    \centering
    \includegraphics{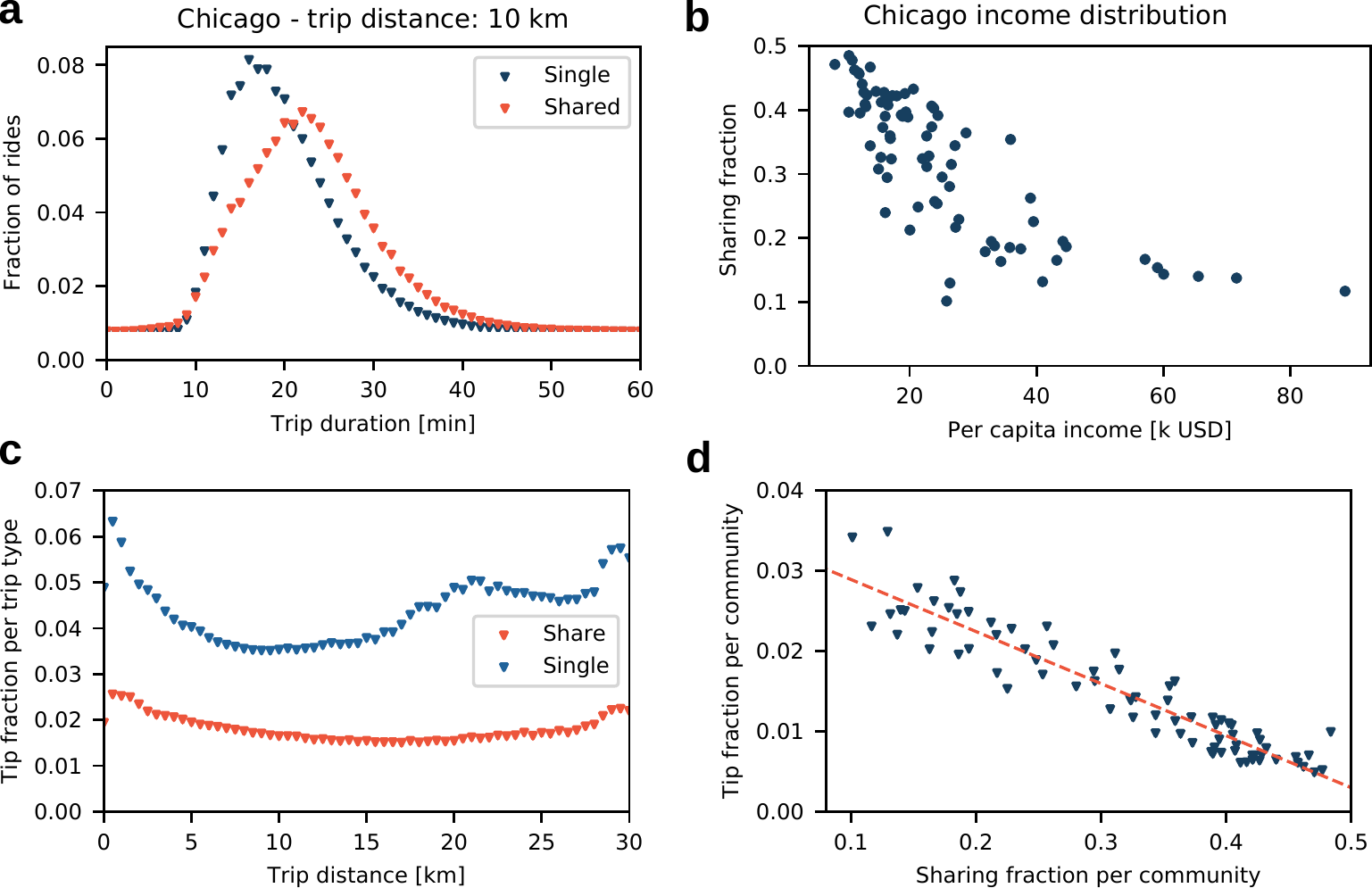}
    \caption{
    \textbf{In Chicago shared rides on average take longer and adoption is correlated with financial observables.}  \textbf{a} Shared rides on average take longer and have higher trip duration variance compared to their same distance single ride  equivalents [mean trip duration for 9.75-10.25 km single rides of 20.1 mins, compared to 23.0 mins for corresponding shared rides]. \textbf{b} In Chicago, sharing fraction and per capita income are negatively correlated [income data based on \cite{Chicago_Income_2020}]. \textbf{c} Across trip types of different distances, shared ride users consistently tip less than single ride users in Chicago, independent of community area. \textbf{d} The average tip fraction per community area is negatively correlated with its sharing fraction.
    }
    \label{fig:SFIG_CHI_TripChars}
\end{figure}

\paragraph{Detour and trip duration uncertainty.} Supplementary Figure \ref{fig:SFIG_CHI_TripChars}a shows trip duration histograms for single and shared rides in Chicago within a given distance interval. Similar to New York City (see Supplementary Fig.~\ref{fig:SFIG_NYC_TripChars} (left)) the trip duration distribution for single rides is positively skewed while it becomes more symmetric for shared rides. The comparison of mean trip durations unveils that single rides are quicker on average (20.1 mins for single rides from all trips with a distance between 9.75-10.25 km, 23.0 mins for corresponding shared rides), and exhibit less trip duration variance.

Ref.~\cite{Schwieterman2019} finds that an average ride-hailing user in Chicago spends approximately 42-108 USD per hour saved when making the decision to use public transit, or ride-hailing services. In that line, the observed average trip duration disadvantage of shared rides may be estimated in the range of 0.7-1.8 USD lost per additional trip minute of a shared ride, assuming linear scaling and similar financial trade-offs as in public transit, thus providing a quantitative estimate for the disutility of ride-sharing. Uncertainty about potential extra time of the trip may add to this estimate at the time of decision-making to book a shared ride or not, e.g. if risk-averse individuals try to anticipate the disutility of worst case detour scenarios.\\

\paragraph{Financial incentives.} Per capita income statistics per community area in Chicago and the corresponding ride-sharing adoption reveal a negative correlation (see Supplementary Fig.~\ref{fig:SFIG_CHI_TripChars}b). Trips originating from community zones with lower per capita income exhibit higher ride-sharing adoption than high income zones, hinting at financial motivation governing ride-sharing decisions (as reported in empirical studies and focus group interviews \cite{Morris2019,Lippke2020}). For origins characterized by incomes exceeding sixty thousand USD the sharing fraction appears to become approximately constant, providing a consistent picture when compared to New York City (see Supplementary Fig.~\ref{fig:SFIG_NYC_TripChars} (right)).\\

\paragraph{Inconvenience.} Tipping is generally perceived as a financial measure for a person's willingness to reward a socially enjoyable experience. Here, we assess the tip fraction -- the ratio of tip to total trip fare -- for single and shared rides as an indirect measure for the perceived convenience of the ride experience (as well as additional evidence for the importance of financial incentives).  Supplementary Figure \ref{fig:SFIG_CHI_TripChars}c shows that ride-hailing users consistently tip less for shared rides in Chicago (across origin zones and anticipated per capita incomes), independent of trip distance. This observation hints at less convenient ride-sharing experiences compared to single rides. Users may try to partially compensate the perceived disutility against additional financial savings. Supplementary Figure \ref{fig:SFIG_CHI_TripChars}d provides a consistent picture of a negative correlation between tip and sharing fraction.

\clearpage
\section{Supplementary Note 3. Ride-sharing anti-coordination game on networks} 

In this Supplementary Note we formally define 
the ride-sharing anti-coordination game introduced in the Main Manuscript. 
We introduce a replicator dynamics governing the evolution of the population's willingness to share their rides. The resulting network dynamics unveils qualitatively different regimes of ride-sharing adoption and spatially heterogeneous sharing patterns, emerging from symmetry breaking.

\subsection{Urban environment}

Denote by $G = (V,E)$ a mathematical graph of an urban street network composed of a node set $V$ and an edge set $E$. Nodes can be identified with individual intersection, census tracts or qualitatively similar zones embedded in space. Edges correspond to streets connecting the different zones and are weighted by the geographical distance between them. The distance 
matrix $D$ bundles the pairwise (shortest path) distances. 
In the following we consider a one-to-many setting where $S$ people request transportation from a single origin $o\in V$ to a destination $d\in V \backslash \{o\}$ on $G$.

\subsection{Replicator dynamics}

Per destination node $d\in V\backslash\{o\}$ the probability $\pi(d,t) \in [0,1]$ defines the local population's ride-sharing adoption when embarking from origin $o$ at time $t$. $\pi(d,t)$ is an aggregate measure for people's ride-sharing willingness, 
describing the average ride-sharing behavior of people with the same origin-destination combination.
The ride-sharing adoption evolves under discrete-time replicator dynamics
\begin{align}
    \pi(d,t+1) &= r(d,t)\pi(d,t)\label{eqn:ReplicatorDynamics}
\end{align}
with reproduction factor
\begin{align}
    r(d,t)&=\frac{E[u_\mathrm{Share}(d,t)]}{E[u(d,t)]}
\end{align}
with the expected utility of sharing $E[u_\mathrm{Share}(d,t)]$ and the population average utility $E[u(d,t)]$ from both shared and single rides. For simplicity, we assume the utility derived from single rides to be constant. Hence, $E[u_\mathrm{single}(d,t)] = u_\mathrm{single}(d) > 0$, allowing to express $E[u_\mathrm{Share}(d,t)] = u_\mathrm{single}(d) + E[\Delta u(d,t)]$. With this shorthand notation the reproduction factor becomes
\begin{align}
    r(d,t)&= \frac{u_\mathrm{single}(d)+E[\Delta u(d,t)}{u_\mathrm{single}(d)+\pi(d,t)E[\Delta u(d,t)]}.
\end{align}
In the limit $u_\mathrm{single}(d)\to \infty$, Supplementary Eqn.~\eqref{eqn:ReplicatorDynamics} becomes equivalent to the continuous-time version of the replicator equation \cite{Gaunersdorfer1995}, but does not otherwise change the equilibrium states of the dynamics.

Depending on the equilibrium value of $E[\Delta u(d)^*]$ the dynamics converges to a pure strategy equilibrium $\pi(d)^* = 0, \ \mathrm{if} \ E[\Delta u(d)^*]< 0$, $\pi(d)^* = 1, \ \mathrm{if} \ E[\Delta u(d)^*]> 0$, or a mixed strategy equilibrium $\pi(d)^*\in (0,1) \ \mathrm{if}\ E[\Delta u(d)^*] = 0$. The incremental utility of sharing decomposes into
\begin{align}
    E[\Delta u(d,t)] &= \epsilon d_\mathrm{single}(d)-\xi E[d_\mathrm{det}(d,t)]-\zeta E[d_\mathrm{inc}(d,t)] \label{eqn:ReplicatorDynamics2}
\end{align}
where $d_\mathrm{single}(d)=D_{o,d}$ is the shortest path distance between origin $o$ and destination $d$, $d_\mathrm{det}(d)$ is the detour from sharing for destination $d$ at time $t$ and $d_\mathrm{inc}(d,t)$ is the distance spent together on a shared ride. While the first distance is deterministic, the latter two are stochastic and depend on the decision of other ride-hailing users and the overall demand for shared rides on the network. Hence, they mediate a coupling between destinations on the network. 

A rescaling of Supplementary Eqn.~\eqref{eqn:ReplicatorDynamics2}
\begin{align}
    E[\Delta u(d,t)] &= \epsilon d_\mathrm{single}(d)\left( 1 -\frac{\xi}{\epsilon} \frac{E[d_\mathrm{det}(d,t)]}{d_\mathrm{single}(d)}-\frac{\zeta}{\epsilon} \frac{E[d_\mathrm{inc}(d,t)]}{d_\mathrm{single}(d)} \right) 
\end{align}
shows that the dimensionless parameters $\xi/\epsilon$ and $\zeta/\epsilon$ as well as the relative detour $E[d_\mathrm{det}(d,t)]/d_\mathrm{single}(d)$ and inconvenience $E[d_\mathrm{inc}(d,t)]/d_\mathrm{single}(d)$ at time $t$ govern whether $E[\Delta u(d,t)]$ is positive or negative. The quantities measure the trade-off between inconvenience or detour to financial discount of a shared ride, respectively. If the financial discount is sufficiently high to compensate for the detours and inconvenience, the replicator dynamics in Supplementary Eqn.~\eqref{eqn:ReplicatorDynamics} will amplify the local population's adoption of ride-sharing. 

\subsection{Expected detour and inconvenience}

The expected detour and inconvenience of shared rides originating from origin $o\in V$, going to destination $d\in V$, depend on (i) the configuration of  destinations in the request set $S$ at time $t$, (ii) the realization of sharing choices across all users, and (iii) the service provider's matching and routing algorithm.

\begin{enumerate}
    \item[($i$)] \textit{Origin-destination distribution}. Denote by $\sigma \in V^S$ the destination request configuration of the $S$ simultaneous transportation requests from $o$. $\sigma$ is a random variable governed by the origin-destination distribution $W_o$. It impacts where users travel and which users may potentially be matched when sharing a ride.
    \item[($ii$)] \textit{Adoption of ride-sharing}. Depending on the user's individual decisions to share their rides, $\sigma$ decomposes into $\sigma_\mathrm{Share}$ and $\sigma_\mathrm{single}$. The realization of destinations in $\sigma_\mathrm{Share}$ determines the potentially shareable rides.
    \item[($iii$)] \textit{Matching and routing algorithm}. Providers match ride requests based on distance savings potentials, which is equivalent to a maximum weight matching problem on a mathematical graph: Shared ride requests define the nodes of this graph. If two rides offer a distance savings potential to the provider compared to two single rides, the ride requests are connected by an edge (see Supplementary Fig.~\ref{fig:MatchingAlgorithm}a). The distance savings potential defines the edge weight. Here, we assume that both for single and shared ride requests the provider needs to return to the trip origin, consistent with the one-to-many setting. The provider's matching algorithm determines the matching of shared ride requests that maximizes the saved distance (see Supplementary Fig.~\ref{fig:MatchingAlgorithm}b). This matching scheme is consistent with a profit-maximizing provider given that there are sufficiently many drivers to serve all customers.
    
    Per matched request pair, the provider defines the trip route to minimize the distance driven. If he is indifferent whom to drop first, he will deliver the passenger with the shorter distance first to minimize customer inconveniences (see Supplementary Fig.~\ref{fig:MatchingAlgorithm}c). If, again, he is indifferent he tosses a fair coin to determine the order of the shared ride.
\end{enumerate}

\begin{figure}
    \centering
    \includegraphics{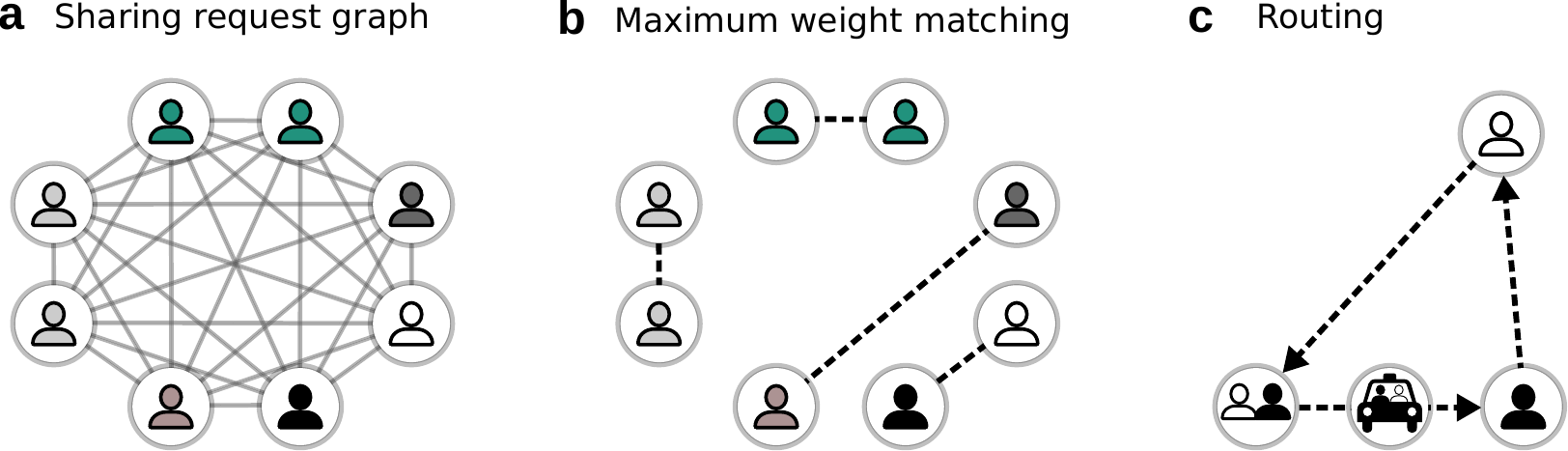}
    \caption{\textbf{Pairing rides is a maximum weight matching problem}. The provider's matching algorithm solves the maximum weight matching problem. \textbf{a} A shared ride request graph defines potentially shareable rides $\sigma_\mathrm{Share}$ with edge weights defining the saved distance of a combined ride. \textbf{b} The provider pairs requests to maximize his saved distance. \textbf{c} Per matching the provider defines the shared route to minimize the distance driven and customer inconvenience.}
    \label{fig:MatchingAlgorithm}
\end{figure}

While the adoption of ride-sharing in general depends on the underlying street network and the destination distribution, the problem simplifies in the limit of many concurrent users, $S \rightarrow \infty.$ In particular, a necessary and sufficient condition for full adoption of ride-sharing to be a stable equilibrium in this limit is that the financial incentives compensate the inconvenience (see Supplementary Fig.~\ref{fig:SFIG7}). In this limit and with full sharing, detours disappear as users will always be matched with other users with the same destination. Formally:

\newpage

\begin{theorem}[Full sharing in high-demand limit]
If $\lim_{S\to\infty} \pi(d)^*=1$ the ratio of inconvenience to financial incentive must be $\zeta/\epsilon<1$.
\end{theorem}
\begin{proof}
A dominant equilibrium strategy in sharing, $\pi(d)^*=1$, implies positive expected utility difference $E[\Delta u(d)^*]>0$. The limit of infinite request number yields
\begin{align}
    \lim_{S\to\infty} E[\Delta u(d)^*] =  d_\mathrm{single}(d)\epsilon \left(1-\lim_{S\to\infty}\frac{\xi}{\epsilon}\frac{E[d_\mathrm{det}(d)]}{d_\mathrm{single}(d)}-\lim_{S\to\infty} \frac{\zeta}{\epsilon}\frac{E[d_\mathrm{inc}(d)]}{d_\mathrm{single}(d)} \right) = d_\mathrm{single}(d)\epsilon \left(1- \frac{\zeta}{\epsilon} \right)
\end{align}
where we used that $\pi(d)^*=1$ for which $S\to\infty$ corresponds to zero-detour matching to destination $d$. Consequently, $E[d_\mathrm{inc}(d)]=d_\mathrm{single}(d)$ and
\begin{align}
    0<E[\Delta u(d)^*] = d_\mathrm{single}(d) \epsilon \left(1-\frac{\zeta}{\epsilon}\right)
\end{align}
which implies $\frac{\zeta}{\epsilon} < 1$.

\begin{figure}
    \centering
    \includegraphics{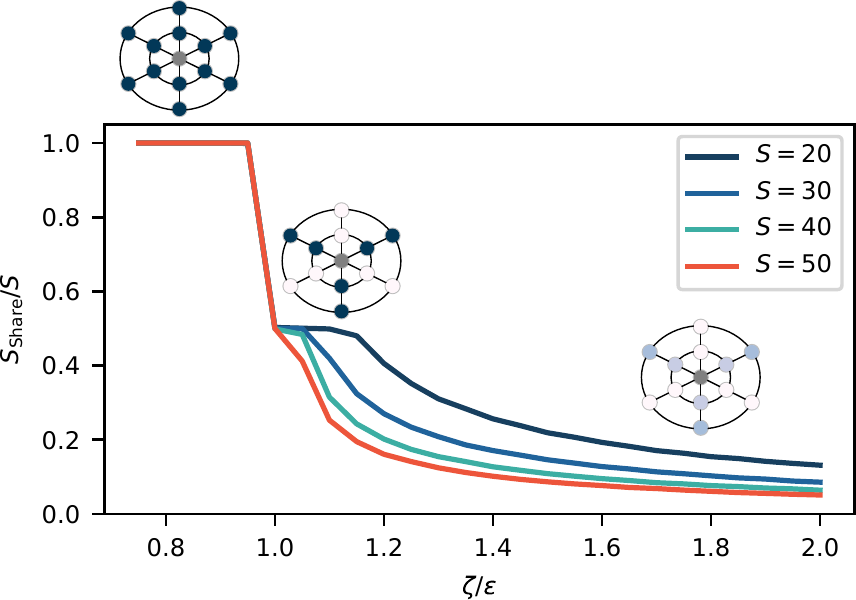}
    \caption{\textbf{Discontinuous phase transition in the ride-sharing adoption in the high-demand limit}. The control parameter $\zeta/\epsilon$ governs the ride-sharing adoption. For $\zeta/\epsilon <1$ the system finds itself in the high-sharing regime ($S_\mathrm{Share}/S=1$). At $\zeta_c/\epsilon_c=1$, a spatially heterogeneous pattern of ride-sharing adoption forms for finite $S$ and detour preferences $\xi>0$ where adjacent branches alternate between the low- and high-sharing regimes. For $\zeta/\epsilon>1$ ride-sharing adoption fades out along the branches that are still in the high-sharing regime. This hybrid state of ride-sharing adoption becomes smaller as $S$ increases. In the limit $S\to\infty$ the transition becomes discontinuous at $\zeta_c/\epsilon_c=1$.
    Simulation parameters: $\xi = 0.3, \epsilon = 0.2$.}
    \label{fig:SFIG7}
\end{figure}

\end{proof}

\subsubsection{The case $S=3$}

The case $S=3$ produces equilibrium adoption of ride-sharing qualitatively different than adjacent values of $S$, as discussed in Figure 3a in the the Main Manuscript. For sufficiently high $\epsilon$ the population has a dominant sharing strategy in this configuration which is induced by the fact that the service provider can at most pair two of the three ride requests into a shared one. The left-over request will enjoy the benefits of a single ride at discounted trip fare, inducing an incentive to become this request which fuels both the ride-sharing adoption as well as the matching probability. 

As $S$ increases beyond $S=3$ the incentive of gambling on being a left-over request reduces drastically as far less corresponding constellations exist. Thus, $S=3$ produces a behavior in Figure 3a of the Main Manuscript that looks qualitatively different than for other values of $S$.

\clearpage

\section{Supplementary Note 4. Robustness of ride-sharing adoption regimes}

In the Main Manuscript we demonstrated that the ride-sharing anti-coordination game reproduces opposing regimes of ride-sharing adoption in a simple setting. In this section we demonstrate the robustness of these results under different conditions, including non-homogeneous demand constellations and for different origin locations in the network, illustrating that the underlying mechanisms balancing incentives remain identical.\\

\subsection*{Ride-sharing adoption for non-homogeneous origin-destination demand}

Using the stylized city topology introduced in the Main Manuscript, we investigate the impact of radially and azimuthally asymmetric destination demand on the ride-sharing adoption from a joint origin. We distinguish between four scenarios representative for different types of urban settlements:
\begin{enumerate}
    \item \textit{Dense core}: Starting from a joint origin in the city center, a gradient of decreasing destination demand in radial direction mimics urban environments with densely populated city core. Further distance destinations (e.g. suburbs) are less often requested, e.g. because of sparser population density.
    \item \textit{Urban sprawl}: In situations where distant destinations from the city center make up the majority of ride requests the radial destination demand gradient is reversed. Theses scenarios represent constellations of urban sprawl, or situations where the city core is only sparely populated, e.g. because of high real-estate prices.
    \item \textit{Sparse settlement}: Urban environments may exhibit azimuthal gradients in destination demand starting from an origin in the city center, e.g. stretched out residential settlements that have formed next to existing road, river banks etc. In that case destination demands in radial direction might be similar, but differ significantly by cardinal direction.
    \item \textit{Heterogeneous settlement}: Urban constellations where both radial as well as azimuthal destination demand gradients exist might describe heterogeneously grown environments, e.g. because of natural obstacles or staged development.
\end{enumerate}

Supplementary Figs.~\ref{fig:RadialsODSharingPattern} and \ref{fig:AzimuthalODSharingPattern} correspond to the four scenarios. For given financial discount $\epsilon$ an increase in request rate $S$ gives rise to a spatially heterogeneous sharing/non-sharing pattern and decreasing overall adoption of ride-sharing in all scenarios, independent of the destination demand distributions (compare Fig.~3 in the Main Manuscript). 
As second order effects, the origin-destination distribution determines (i) whether the cardinal direction of the sharing pattern is random (Supplementary Fig.~\ref{fig:RadialsODSharingPattern} for radially asymmetric destination demand), or aligned with the highest destination demand (Supplementary Fig.~\ref{fig:AzimuthalODSharingPattern} for azimuthally asymmetric destination demand), and (ii) whether close-by or distant destinations reduce their willingness to share first upon increased request rate.
\begin{enumerate}
    \item \textit{Dense core}: For \textit{dense core} settings (see Supplementary Fig.~\ref{fig:RadialsODSharingPattern}a) the cardinal direction of the sharing/non-sharing pattern is solely driven by random fluctuations breaking the azimuthal symmetry. The destination demand gradient leads to a reduction of willingness to share from inside to outside as $S$ increases.
    
    \item \textit{Urban sprawl}: Phenomenologically, \textit{urban sprawl} (see Supplementary Fig.~\ref{fig:RadialsODSharingPattern}b) corresponds to \textit{dense core}, but this time increasing the request rate reduces the willingness to share from the outside (i.e. high destination demand).
    
    \item \textit{Sparse settlement}: In the presence of azimuthal destination demand gradients the sharing pattern forms along the branches of high demand (see Supplementary Fig.~\ref{fig:AzimuthalODSharingPattern}a). The dominance of those destinations in the replicator dynamics guides the symmetry breaking into letting low demand destinations reduce their willingness to share, which reduces the expected detour for sharing branches. As $S$ increases the willingness to share reduces from in- to outside as in the uniform case analyzed in the Main Manuscript.
    
    \item \textit{Heterogeneous settlement}: For heterogeneous settlements the combination of radial and azimuthal gradients in the destination demand orients the cardinal direction of the sharing pattern in line with high demand outside destinations (see Supplementary Fig.~\ref{fig:AzimuthalODSharingPattern}b). 
    Outside destinations have a higher utility gain of sharing (incentives proportional to distance), resulting in faster adjustments in the replicator dynamics and an effective first-mover advantage such that these destinations determine the spatial pattern in the partial sharing state.
\end{enumerate}
In all settings, the results are qualitatively the same as for homogeneous origin-destination demand (compare Fig.~3 in the Main Manuscript). Naturally, sufficiently high financial incentives overcome this partial-sharing phase and result in full sharing, reproducing the two phases of ride-sharing adoption (see Supplementary Fig.~\ref{fig:SFIG5}, compare Fig.~4 in the Main Manuscript).

\begin{figure}
    \centering
    \includegraphics[width=1.\linewidth]{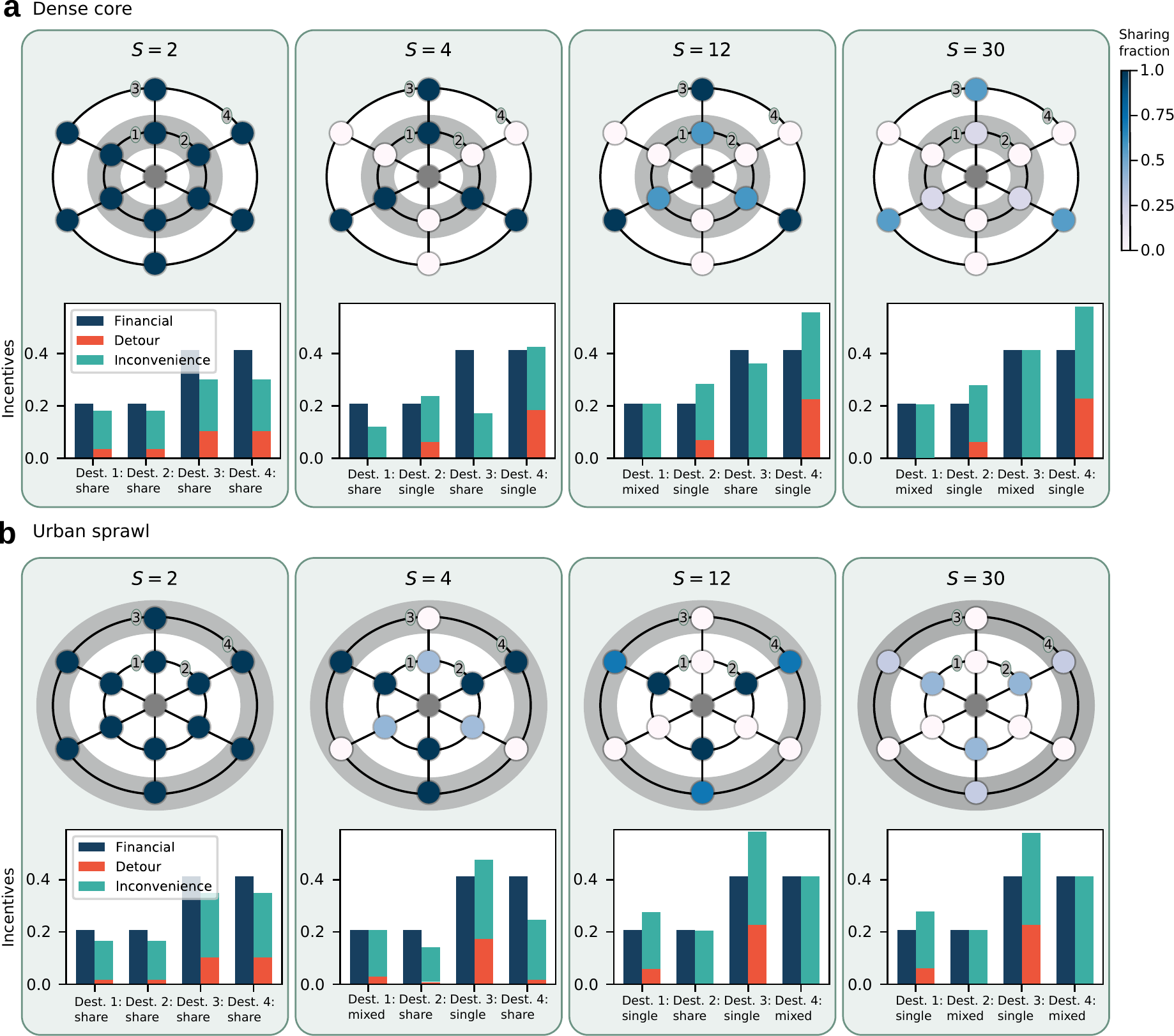}
    \caption{
    \textbf{Robustness of ride-sharing adoption for radially asymmetric origin-destination demand}. Radial asymmetry of the destination demand distribution does not qualitatively affect the equilibrium ride-sharing adoption. 
    \textbf{a} Dense core setting, inner ring destinations (gray shading) are visited twice as often as outer ring destinations. Increased request rate reduces the destination's ride-sharing adoption from inside to outside. 
    \textbf{b} Urban sprawl setting, outer ring destinations (gray shading) are visited twice as often as inner ring destinations. As $S$ increases the ride-sharing adoption ceases from outside to inside. 
    In both cases branches of high ride-sharing adoption emerge in a random direction (compare Fig.~3 in the Main Manuscript). Parameters: $\epsilon = 0.2, \zeta=0.3, \xi=0.3$.}
    \label{fig:RadialsODSharingPattern}
\end{figure}

\clearpage

\begin{figure}
    \centering
    \includegraphics[width=1.\linewidth]{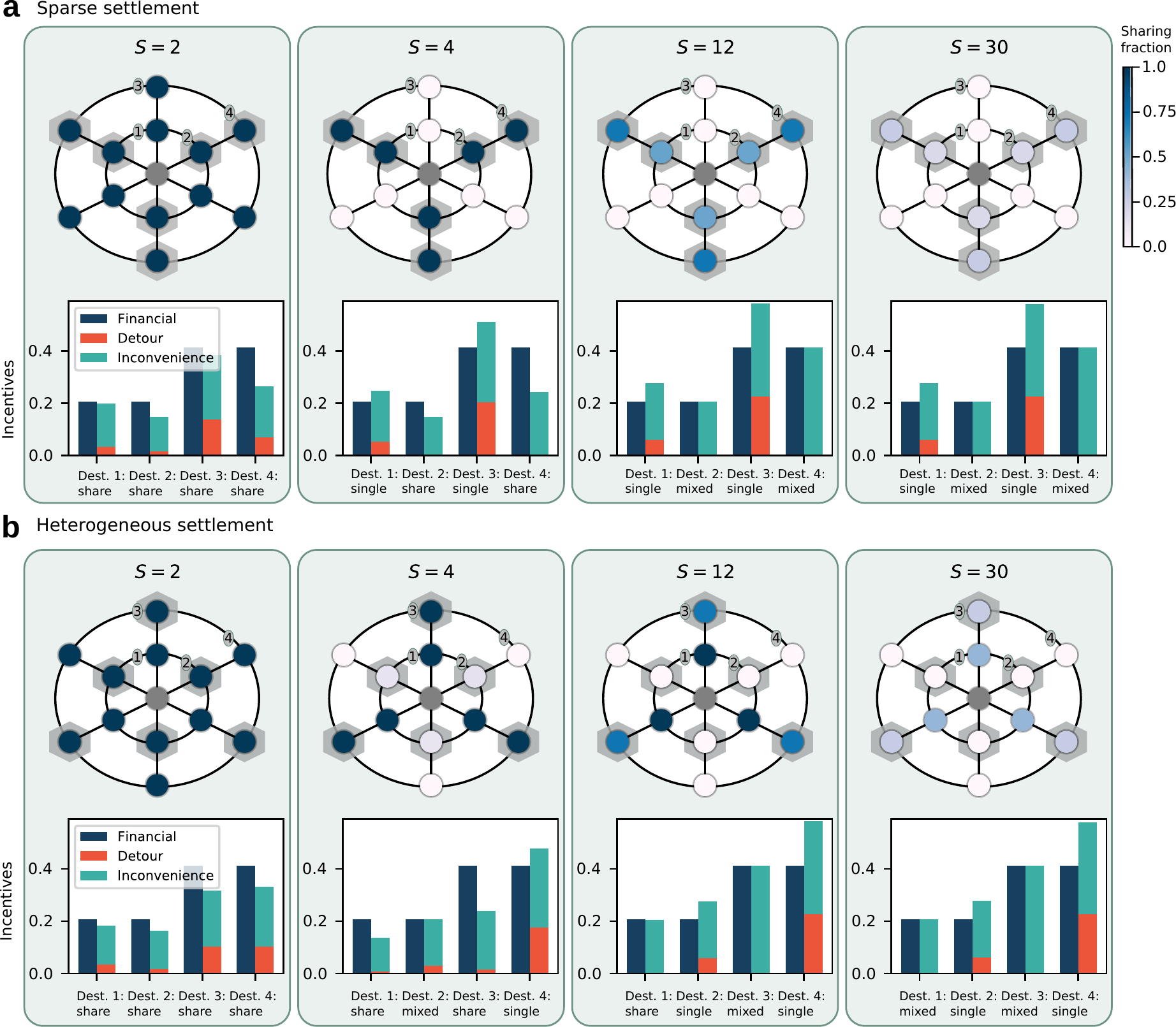}
    \caption{
    \textbf{Robustness of ride-sharing adoption for azimuthally asymmetric origin-destination demand}. 
    An azimuthally asymmetric destination demand distribution predetermines the emergence of sharing/non-sharing branches.
    \textbf{a} Sparse settlement setting, neighboring branches of destinations in radial direction alternate between being visited twice as likely (gray shading) as the other branches. The high-demand branches are sharing while the low demand branches quit sharing due to their high expected detour. Increased request rate reduces the destination's ride-sharing adoption from inside to outside. 
    \textbf{b} Heterogeneous settlement setting, inner and outer ring destination nodes on the same branch alternate between being requested twice as often (gray shading) as the other one. Also in this setting branches of high ride-sharing adoption emerge, driven by the outermost destinations. As $S$ increases the ride-sharing adoption ceases from outside to inside (compare Supplementary Fig.~\ref{fig:RadialsODSharingPattern}b). 
    Parameters: $\epsilon = 0.2, \zeta=0.3, \xi=0.3$.
    }
    \label{fig:AzimuthalODSharingPattern}
\end{figure}

\clearpage

\begin{figure}
    \centering
    \includegraphics{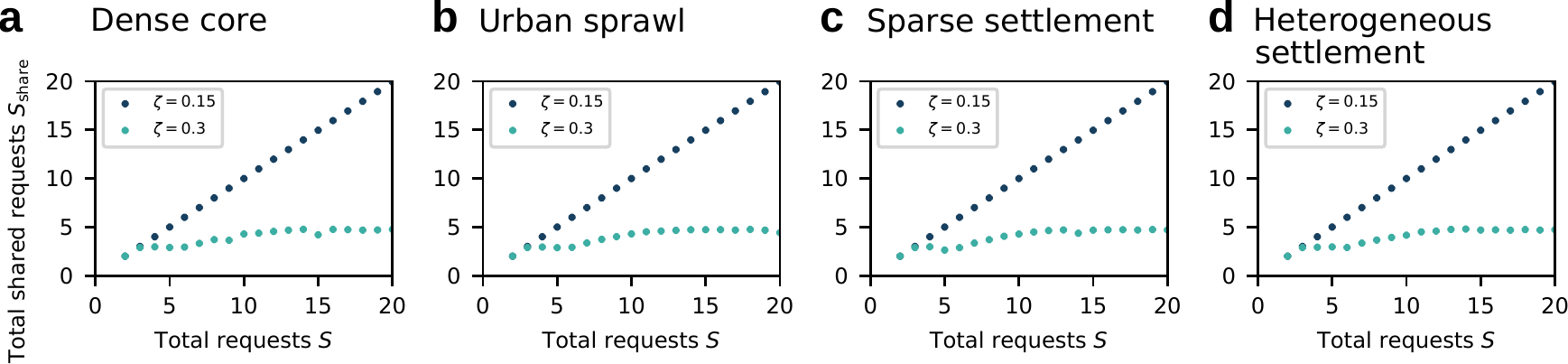}
    \caption{{\textbf{Phases of ride-sharing adoption persist for heterogeneous demand distributions}}. All settings of heterogeneous demand constellation (see Supplementary Note 4) reproduce the low-sharing phase ($S_\mathrm{Share} = \mathrm{const}$ [green]) and high-sharing phase ($S_\mathrm{Share} = S$ [blue]) illustrated in the Main Manuscript for financial incentives that (do not) compensate the inconvenience (compare Fig.~4b in the Main Manuscript).}
\label{fig:SFIG5}
\end{figure}

\clearpage
\subsection{Ride-sharing adoption for decentral origin}

In the one-to-many ride-sharing game, the relative position of the origin defines the scale of average distances to different destinations and the possible combinations in which requests for shared rides are matched. Here, we consider the stylized city topology introduced in the Main Manuscript with a decentral origin at the periphery. The decentral position of the origin breaks the radial symmetry of the city layout and creates a heterogeneous distribution of trip distances, reminiscent of real road networks. Consequently, it impacts expected detours and inconvenience. We demonstrate the robustness of the different phases of ride-sharing adoption even under these conditions.

\begin{figure}[h]
    \centering
    \includegraphics{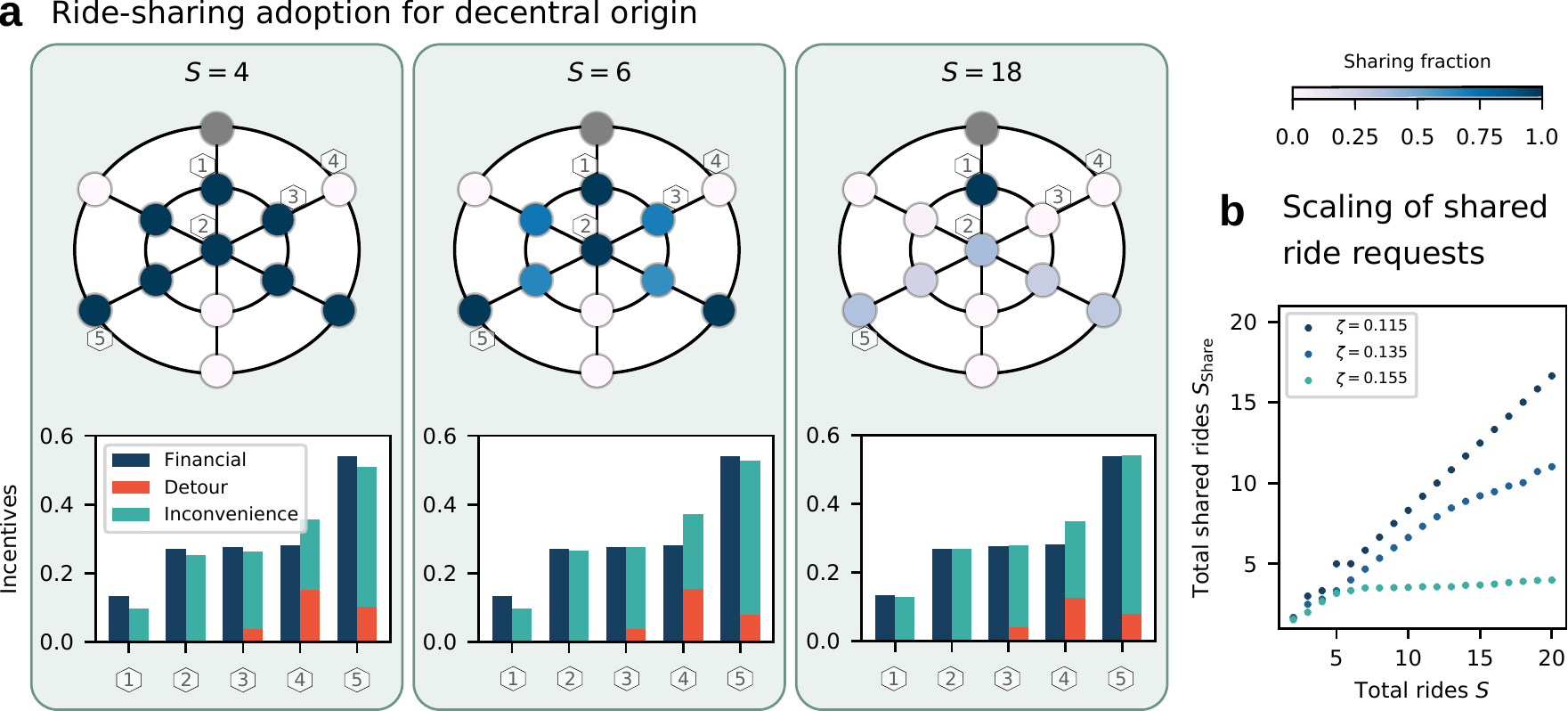}
    \caption{\textbf{Robustness of ride-sharing adoption for decentral origin}. A decentral origin node (gray) defines a new scale of average distances to the destinations in the stylized city topology, but the ride-sharing adoption remains qualitatively the same. \textbf{a} In the partial sharing regime a spatially heterogeneous pattern emerges in the ride-sharing adoption across origins, and fades out as the number of rides $S$ increases (Parameters: $\epsilon = 0.135, \zeta = 0.155, \xi = 0.3$). \textbf{b} If financial incentives $\epsilon$ overcompensate inconvenience effects $\zeta$ we recover a full-sharing phase (Parameters: $\epsilon=0.135, \xi = 0.3$).}
    \label{fig:SFIG6}
\end{figure}

Supplementary Figure~\ref{fig:SFIG6} illustrates a one-to-many situation of homogeneous transportation demand from the northernmost node in the stylized city topology. Again, a distinct spatial pattern of ride-sharing adoption emerges in the partial sharing regime (Supplementary Fig.~\ref{fig:SFIG6}a). When the financial incentives are sufficiently large, we recover the full sharing phase (Supplementary Fig.~\ref{fig:SFIG6}b). In this setting, the sharing pattern is symmetric about the north-to-south axis, where nodes in the direction of the city center share dominantly. They have no expected detours since in all constellations where they are matched, they will be dropped first. This is not the case anymore for destinations on the opposite side of the city center. Hence, these destinations do not share. For the remaining destinations the decision to share, or not, results in a zero-sum game very soon as $S$ increases (compare Supplementary Fig.~\ref{fig:SFIG6}a, center panel) and eventually reproduces a ride-sharing adoption pattern where neighboring branches alternate between sharing and not sharing (compare Fig.~3 in Main Manuscript). 

Again, ride-sharing adoption behaves qualitatively similar compared to the constellation for central origins.

\newpage

\subsection{Ride-sharing adoption under imperfect information}

Under realistic settings ride-hailing users may not be able to perfectly quantify or monitor the utility of shared rides. For example, users cannot exhaustively sample all demand configurations and sharing decisions by other users and may update their sharing decisions based on limited observations from, say, their past week of usage. Additionally, external shocks may modify the realized utility per time step, or individuals may have different utility perception, even though they travel from the same origin to the same destination. In sum, ride-hailing users may exhibit utility fluctuations and may make imperfect decisions based on these. We demonstrate the robustness of our results with respect to uncertainty and fluctuations by adding stochastic noise to the realized utility of shared rides.

Taking $X\sim\mathcal{N}(0,\sigma^2)$ to be a normal random variable with mean zero and variance $\sigma^2$, we realize the deterministic utility increment $\Delta u$ and multiply it by a factor $(1+x_{i,n,t})$ where $x_{i,n,t}$ denotes the (independent and identically distributed) realization of random variable $X$ for rider $i$ in game $n$ at time $t$. The standard deviation $\sigma$ acts as a control parameter for how different an individual might perceive the utility of a shared ride, or how strong external stochastic influences are. Supplementary Figure \ref{fig:SFIG_UtilityFluct} illustrates the evolution and average equilibrium ride-sharing adoption for different destination nodes in the stylized street network topology for various degrees of imperfect information $\sigma$. 

The ride-sharing adoption evolves based on the estimated expected utility from $100$ realizations per replicator step [see Supplementary Fig.~\ref{fig:SFIG_UtilityFluct} (bottom)], e.g. the experience of a group of similar users over a week. As $\sigma$ increases, the trajectories become subject to stronger fluctuations over time, but the average equilibrium results remain unaltered even for large values of $\sigma$ [compare Supplementary Fig.~\ref{fig:SFIG_UtilityFluct} (top)]. Hence, users naturally filter out large variances even though they may not arrive at a perfect estimate of their expected utility. 

\begin{figure}[h]
    \centering
    \includegraphics{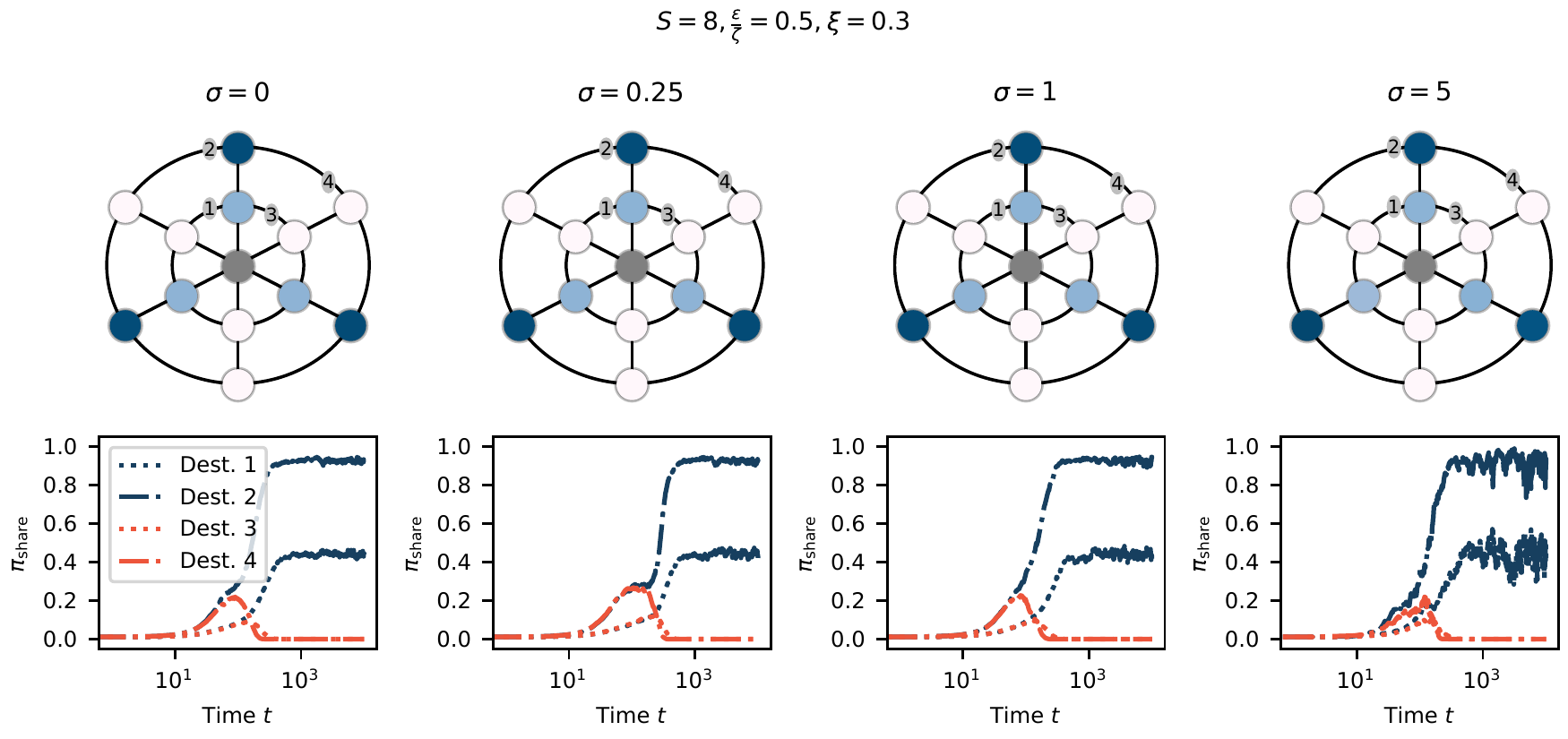}
    \caption{\textbf{Robustness of ride-sharing adoption under utility fluctuations, reflecting imperfect information and heterogeneous utility perception}. (Top) Average equilibrium ride-sharing adoption under fluctuations in the perceived utility of shared rides $\Delta u$, multiplied by a stochastic factor governed by a standard normal distribution with standard deviation $\sigma$. (Bottom) Trajectories of the replicator dynamics under stochastic utility ($u_\mathrm{single}=3$). Persistently, a dynamic symmetry breaking process leads to adoption of opposing ride-sharing strategies along adjacent destination nodes in cardinal direction.}
    \label{fig:SFIG_UtilityFluct}
\end{figure}

\newpage
\subsection{Ride-sharing adoption in systems with heterogeneous convenience preferences}

In real systems preferences of players participating in the ride-sharing game might be heterogeneous, impacting the macroscopic state of ride-sharing adoption. Here, we consider heterogeneous convenience preferences $\zeta$ among the $S$ people starting their ride-hailing or -sharing trips from a common origin, and demonstrate the robustness of the different phases of ride-sharing adoption.

Given fixed financial discounts $\epsilon= 0.2$, players belong to one of three convenience types $\zeta\in\{0.175,0.225,0.275\}$. Per game play the composition of the $S$ players is drawn randomly from a distribution with Pr$[\zeta_i]=p_i,\ i\in\{1,2,3\}$ where $p_1+p_2+p_3=1$. Note that the distribution of convenience preferences tunes the average convenience preferences $E[\zeta]=\zeta_1p_1+\zeta_2p_2+\zeta_3p_3$.

Supplementary Figure~\ref{fig:HetZeta_SShare} (left) illustrates the sharing adoption under increasing demand for balanced distribution of convenience preferences. At low demand $S$ all players find themselves in the high-sharing state to maximize their personal utility ($S_\mathrm{Share}=S$). As $S$ increases further, $S_\mathrm{Share}=\alpha S$ with $\alpha <1$, hinting at a partial-sharing regime. In fact, $\alpha$ approaches a value of $1-$Pr[$\zeta_3]=1-p_3$ when increasing $S$. The player type with the highest convenience preferences ($\zeta_3$) quits sharing first as financial discounts cannot compensate their expected inconvenience. Hence, the system may appear to be in a partial sharing state macroscopically, but finds itself in a hybrid state of high- and low-sharing adoption microscopically. For even higher demand, the second player type ($\zeta_2$) will also quit sharing, explaining why $S_\mathrm{Share}\approx (1-p_2-p_3)S$ for very high demand $S$. For the first player type ($\zeta_1$) financial discounts will always overcompensate expected inconvenience effects. Thus, these players adopt a dominant sharing strategy and $S_\mathrm{Share}\approx p_1 S$ in the high demand limit.

In constellations where $\zeta/\epsilon > 1$ for the entire population (i.e. $p_1=0$, see Supplementary Fig.~\ref{fig:HetZeta_SShare} (right)) the system approaches the low-sharing state with constant $S_\mathrm{Share}$ if $S$ is sufficiently large. 

Consequently, heterogeneous distributions of convenience preferences with proportions of the population exceeding $\zeta/\epsilon>1$ may explain an effective sublinear trend of $S_\mathrm{Share}$ as demand $S$ increases.

\begin{figure}[h]
    \centering
    \includegraphics{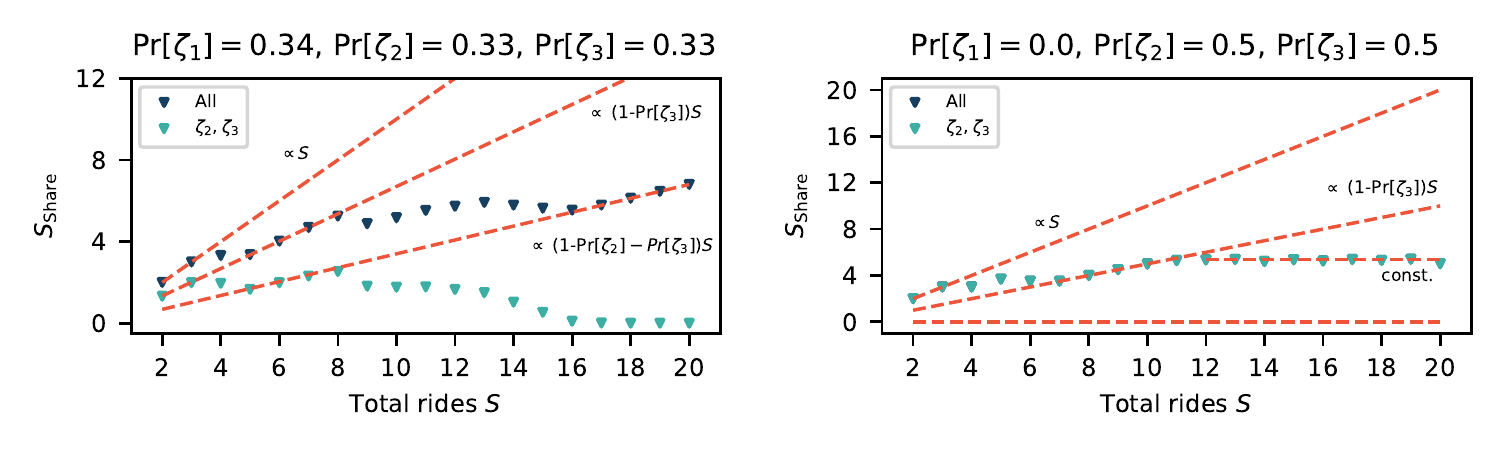}
    \caption{\textbf{Heterogeneity in convenience preferences leads to superposition of low-, partial and high-sharing regimes}. (Left) Uniform distribution of convenience preferences $\zeta_i\in\{0.175,0.225,0.275\}$ transitions from high-sharing state ($S=2-3$) to hybrid state in low- and high ride-sharing adoption ($S>16$) as overall demand $S$ increases. For given financial discount $\epsilon=0.2$ only users with preferences $\epsilon/\zeta_1\approx 1.14$ adopt sharing in the high-demand limit, leading to $S_\mathrm{Share}\approx \mathrm{Pr}[\zeta_1] S$ in the high demand limit. (Right) If $\epsilon/\zeta_i<1$ for all user groups, the system asymptotically approaches the low-sharing regime with a constant number of shared ride requests at high demand.}
    \label{fig:HetZeta_SShare}
\end{figure}

\newpage
\subsection{Superposition of origins with heterogeneous convenience preferences}

In a strongly simplified modeling approach, the aggregate state of ride-sharing in real cities may be imagined as a superposition of different origins with characteristic demand $S$ for rides, and potentially heterogeneous user types with regards to their convenience preferences. For simplicity we neglect heterogeneities in destination demand per origin and topological differences in average trip distances by origin.

Based on the equilibrium results of the ride-sharing game at a single origin with heterogeneous convenience preferences explained in the previous section, we superimpose 600 of them in Supplementary Figure \ref{fig:SFIG_Heteroskedasticity}. Per origin we draw its characteristic demand $S\in[2,20]$ from an exponential probability distribution Pr[$S]\propto e^{-\lambda S},\ \lambda >0$, corresponding to few high demand locations (e.g. airports, downtown) and a large number or low-demand locations (e.g. suburbs). 

For given inconvenience parameters $(\zeta_1=0.175,\zeta_2=0.225,\zeta_3=0.275)$ we draw the distribution $(p_1=\mathrm{Pr}[\zeta_1],p_2=\mathrm{Pr}[\zeta_2],p_3=\mathrm{Pr}[\zeta_3])$ per origin from a normal distribution with mean $\epsilon/(p_1\zeta_1+p_2\zeta_2+p_3\zeta_3) = 1.05$ and standard deviation $0.085$. Note that while, macroscopically, the system appears to be in a full-sharing state \textit{on average}, selected zones may exhibit $\zeta$-distributions for which financial discounts do not compensate inconvenience effects ($\epsilon/\zeta<1$). As explained in the previous section, user types with $\zeta>\epsilon$ may quit sharing at sufficiently high demand $S$.

\begin{figure}[h]
    \centering
    \includegraphics{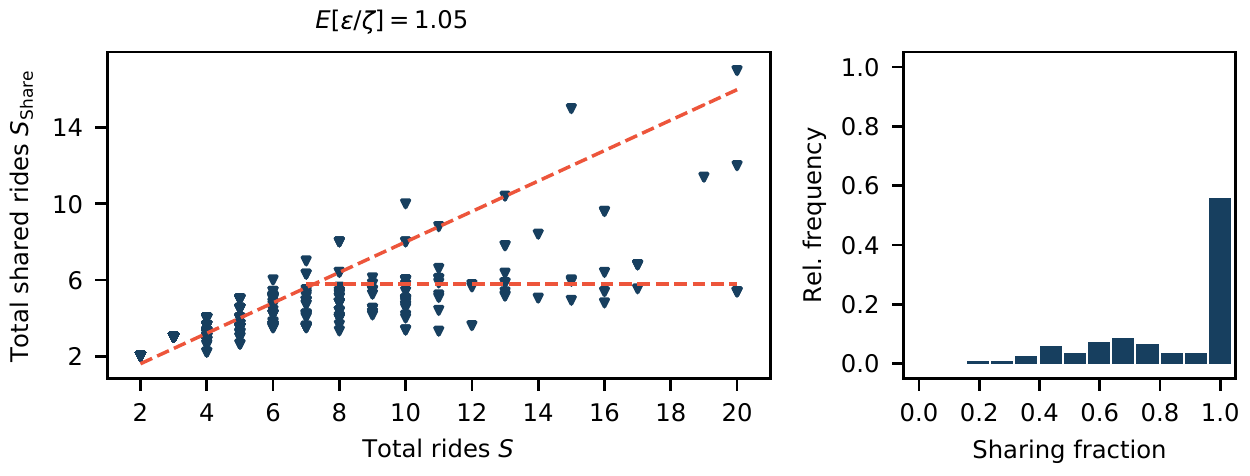}
    \caption{\textbf{Hybrid states of low- and high-sharing adoption from superposition of origins with heterogeneous convenience preferences}. Superimposing origin zones with different distributions of inconvenience preferences can macroscopically give rise to hybrid states of low- and high-sharing adoption at the same time, despite the total average of all zones $E[\epsilon/\zeta]=1.05$ hints at a full-sharing regime. Origin zones for which the vast majority of the local population has higher convenience preferences than current financial incentives accumulate on the horizontal low-sharing branch.}
    \label{fig:SFIG_Heteroskedasticity}
\end{figure}

Supplementary Figure \ref{fig:SFIG_Heteroskedasticity} illustrates the macroscopic state of sharing adoption as a function of demand for 600 superimposed origin zones. Clearly, the system appears to be in a hybrid state of ride-sharing adoption, despite macroscopically the average of financial to inconvenience preferences $E[\epsilon/\zeta]=1.05$ suggest full sharing adoption (but only less than 60\% adopt sharing unconditionally). Due to locally heterogeneous distributions of convenience preferences certain parts of the local population will eventually quit sharing at high demand, pushing those origins into the partial or low-sharing state.

Origins for which the majority of the population has higher convenience preferences than current financial incentives accumulate on the horizontal branch in Supplementary Figure \ref{fig:SFIG_Heteroskedasticity}. The remaining origins for which $p_1\gg p_2+p_3$, form the positive slope branch of $S_\mathrm{Share}$ with increasing demand.

Note that at low demand $S_\mathrm{Share}$ increases with demand $S$ for almost all origins with a wide variety of different slopes, defining a cone-like structure. The opening angle of the cone reflects the variance in $p_3=\mathrm{Pr}[\zeta_3]$-values across the sample, which transition away from the high-sharing state first under increasing demand. 

The resulting macroscopic city view yields a qualitatively similar picture as observed in empirical data for New York City and Chicago (compare Fig.~5 in the Main Manuscript).

\clearpage 

\subsection{Ride-sharing adoption under alternative matching strategies}

Ride-sharing service providers may adopt different matching strategies for shared rides depending on their business rationale. Besides minimizing the total distance driven (compare Main Manuscript and Supplementary Note 3), further strategies may include:
\begin{itemize}
    \item[(i)] \textit{First-in, first-out matching}: The provider matches requests as they come in, irrespective of direction or distance saved, to provide offer the customer a ride as quick as possible, and achieve high vehicle utilization. This strategy corresponds to an unweighted matching problem (see Supplementary Fig.~\ref{fig:MatchingAlgorithm}), and may be interpreted as a real-world scenario where the provider applies a first-in, first-out matching strategy without batch-processing. 
    \item[(ii)] \textit{Batch-processing}: As in the main manuscript, the provider batch-processes all $S$ requests within a specific time-window to pair individuals to optimize platform profits. However, unlike in the main manuscript the provider always matches shared ride requests independent of a distance savings potential. This matching scheme is similar to the one introduced in the Main Manuscript and detailed in Supplementary Note 3, but without the distance constraint on potential matches.
\end{itemize}
Supplementary Figure \ref{fig:FIG_MatchingAlgorithm} demonstrates the robustness of the low-, partial-, and high-sharing regimes of ride-sharing adoption under these alternative matching strategies. For financial discounts insufficient to compensate inconvenience disutilities the number of shared rides becomes constant as the total demand increases (Supplementary Fig.~\ref{fig:FIG_MatchingAlgorithm}, green). Only if financial discounts overcompensate expected detour and inconvenience effects,
\begin{align}
         \epsilon\, d_\mathrm{single}  > \zeta \, E[d_\mathrm{inc}|\mathrm{Share}] + \xi \, E[d_\mathrm{det}|\mathrm{Share}],
\end{align}
such that $E[\Delta u] > 0$ we observe the full-sharing regime where $S_\mathrm{Share}= S$ (Supplementary Fig.~\ref{fig:FIG_MatchingAlgorithm}, blue). In contrast to the distance minimizing matching scheme discussed in the Main Manuscript, the detour preference of customers remains important as long detours may persist, $E[d_\mathrm{det}|\mathrm{Share}] > 0$, despite better matching options in the high demand limit. The partial-sharing regime (Supplementary Fig.~\ref{fig:FIG_MatchingAlgorithm}, orange) separates the other two regimes. These results qualitatively reproduce Fig.~4 from the Main Manuscript.

\begin{figure}[h]
    \centering
    \includegraphics[width=\linewidth]{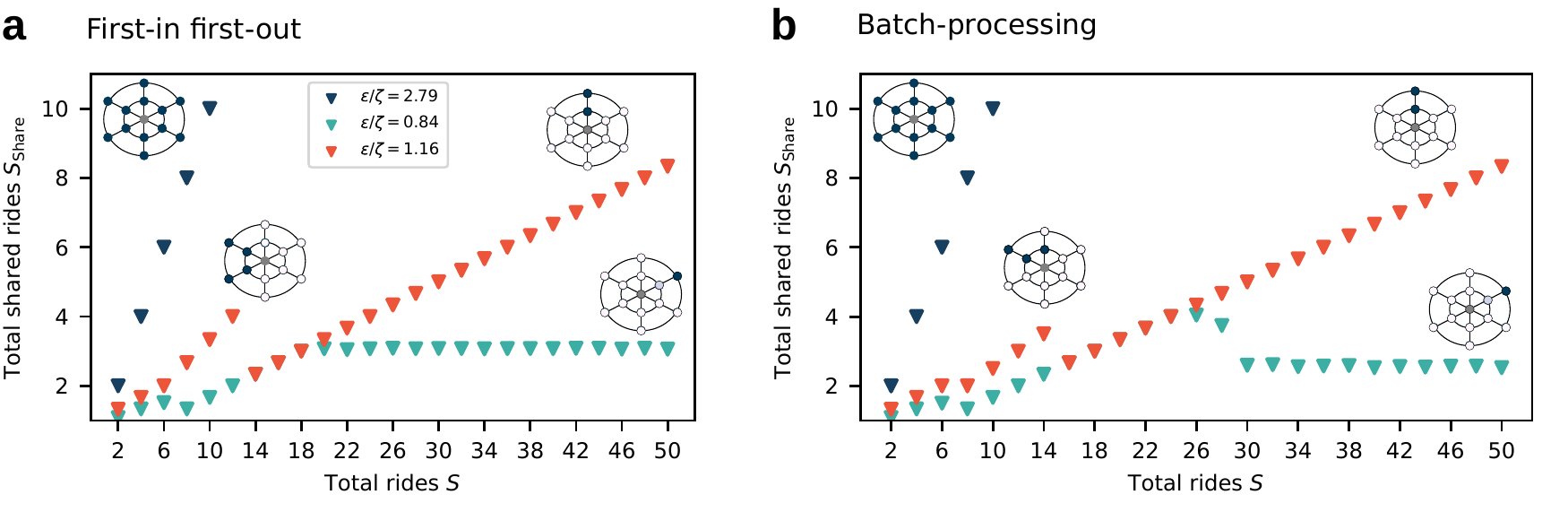}
    \caption{\textbf{Robustness of regimes of ride-sharing adoption for alternative matching strategies}. Matching shared rides independent of distance savings potential compared to single rides reproduces low- (green), partial- (orange) and high-sharing (blue) regimes of ride-sharing adoption. The orientation of spatial ride-sharing patterns results from symmetry breaking, reflecting how ride-sharing users try to minimize expected detour and inconvenience disutilities. \textbf{a} Unweighted matching based on a first-in, first-out principle (pairing of subsequent rides, independent of destination).  \textbf{b} Batch-based matching to optimize platform profits (pairing rides across all destinations, but aiming to maximize platform profits).}
    \label{fig:FIG_MatchingAlgorithm}
\end{figure}

While the spatial patterns of sharing adoption change slightly under alternative matching strategies (see Supplementary Fig.~\ref{fig:FIG_MatchingAlgorithm}), the underlying mechanism leading to the emergence of disparate regimes of ride-sharing adoption robustly persists. Users seek to avoid expected detours and inconvenience effects from ride-sharing. To minimize these disutilities only close-by destination nodes adopt sharing strategies with a spatial extension depending on the financial incentives. A symmetry breaking process selects the sharing destinations (compare insets in Supplementary  Fig.~\ref{fig:FIG_MatchingAlgorithm}). For matching strategies not aimed at minimizing distance driven only one cardinal sharing direction persists under increasing demand, as alternative settings could potentially lead to long detours. Individuals sharing rides in this cardinal direction manage to completely avoid detours. If financial incentives are sufficiently high to compensate for the remaining inconvenience, the partial-sharing regime persists under increasing demand (Supplementary  Fig.~\ref{fig:FIG_MatchingAlgorithm}, orange). If financial incentives are insufficient, the sharing adoption gradually fades out, resulting in the saturation of the number of shared rides (Supplementary Fig.~\ref{fig:FIG_MatchingAlgorithm}, green).

\clearpage
\section{Supplementary Methods}

In this section of the Supplementary Information we provide detailed insight into the data used, cleansing procedures applied and simulation methods implemented.

\subsection{Data acquisition, structure, and treatment}

\paragraph{Data sources.} The New York City Taxi \& Limousine Commission (TLC) publishes trip records for high-volume for-hire vehicles (HVFHV) on a monthly basis. The data includes trip information on pickup time, origin zone, drop-off time, destination zone as well as a shared ride request label for providers completing more than 10000 trips per day \cite{data_NYC}. Our analysis is based on the aggregate HVFHV activity between January and December 2019, independent of service provider, including more than 250 million total trips. We exclude older data due to regulatory changes effective in 2019 \cite{TLC2018}, potentially impacting ride-hailing behavior, and data from 2020 due to changed transportation service activity in the course of the COVID-19 pandemic \cite{WHO2020}.

TLC partitions New York City into 265 taxi zones and provides geospatial information about zone boundaries, names and jurisdictions \cite{data_NYC}. We adopt the definition of these zones in all of our analyses.

Additionally, the City of Chicago publishes ride-hailing trip records on its Open Data Portal \cite{data_Chicago}. The data contains, amongst others, information about trip origin, destination, pickup and dropoff times as well as information whether a shared ride has been authorized by the requester. Our analysis encompasses the time-span between January and December 2019, as chosen for New York City, and includes more than 110 million trip requests served by three transportation service providers (Uber, Lyft, Via).

In our geospatial analysis we restrict ourselves to Chicago's 77 community areas, as well as trips leaving or entering the official city borders.\\

\paragraph{Data preparation.} We use TLC's data as-is. Our data cleansing procedure removes trip records for which trip information is decoded as not available. Furthermore, we omit trip records for zones 264 and 265 in our analysis. While the dataset contains trip requests labeled by these zones, there is no geographic decoding specified by TLC, nor do the zones have names.

Similarly, we use the Chicago trip records as-is.

For our analyses, we determine the total flux matrix specified in Supplementary Eqn.~\eqref{eqn:FluxMatrix} per city. When showing daily averages we normalize the total annual flux between origin and destination zones $o$ and $d$ by 16 hour days to obtain a per-minute-request rate, assuming hardly any request activity for 8 hours per day. In case of specifically defined time windows (see Supplementary Note 1 and 2), we normalize the total flux by the window size. 

We compute the fraction of shared rides as specified in Supplementary Eqn.~\eqref{eqn:FractionSharedRides} using a cutoff threshold of 100 rides for the aggregated records in 2019 to avoid strong statistical fluctuations. Given the total number of hundreds of millions of rides per year in the cities considered in this manuscript, origin-destination zone pairs with less than 100 rides per year typically correspond to trips to destinations with low population density. For example, in New York City in Fig. 1 in the Main Manuscript rides from 'East Village' to 'Governor's Island/Ellis Island/Liberty Island', 'Crotona Park' (a public park in South Bronx), 'Jamaica Bay' (a bay of marshy islands), 'Marine Park/Floyd Bennett Field' (a former airfield), 'Rikers Island' (a jail complex) or 'Saint Michaels Cemetery/Woodside' (a cemetery in Queens) are requested less often than 100 times per year and thus grayed out.

Note that these rides are still included in Figures and analyses illustrating the overall sharing adoption (compare Figs. 5a,b in the Main Manuscript), and only excluded when computing detailed origin-destination resolved sharing fractions (as displayed in Figs. 1, or 5c-f). Figs. 5a,b hint at bifurcation points of approx. 5.5 rides/min in New York City and 2.75 rides/min in Chicago, corresponding to roughly 2 million and 1 million annual ride requests, respectively. Hence, the choice of threshold does not impact the qualitative results up to threshold choice of several hundred thousand rides per year.

\newpage

\subsection{Numerical simulations of ride-sharing anti-coordination games}

\textit{Equilibration}. In this article, we focus on the equilibrium properties of the replicator dynamics underlying the ride-sharing game on networks. To equilibrate the system we evolve Supplementary Eqn.~\eqref{eqn:ReplicatorDynamics} using a value of $u_\mathrm{single}=4$ unless stated otherwise. 

The characteristic timescale of equilibration in the replicator dynamics depends not only on the choice of preference parameters, but also on the order of magnitude of the riders' utilities determined by the value of $u_\mathrm{single}$. The order of magnitude in utility effectively controls the adaptation of $\pi(d,t)$ between successive time steps $t$ and $t+1$, and may be interpreted as the Euler step in the continuous-time formulation of the replicator dynamics. As shown in Supplementary Fig.~\ref{fig:SFIG_uSingle} the qualitative results of the model are robust across a broad range of utility parameters. Hence, the timescale of equilibration may be fine-tuned by appropriate choice of $u_\mathrm{single}$ without altering the equilibrium outcomes. 

\begin{figure}[h]
    \centering
    \includegraphics{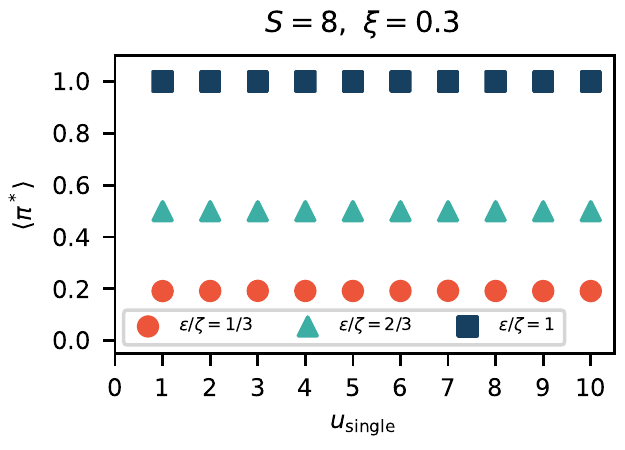}
    \caption{
    \textbf{Robustness of model results across a broad range of utility values of $u_\mathrm{single}$ in all regimes of ride-sharing adoption.} Parameter values for $u_\mathrm{single}$ govern the timescale of equilibration the adaptation between successive steps in the replicator dynamics of the ride-sharing game. The equilibrium outcomes remain unchanged irrespective of phase of ride-sharing adoption across a broad range of values for $u_\mathrm{single}$.
    }
    \label{fig:SFIG_uSingle}
\end{figure}

\newpage

We evolve the system for 20000 replicator time steps, discard a transient of 19000 time steps, quantify the degree of fluctuations around the average of the sharing adoption per destination for the remaining 1000 time steps and test whether it exceeds a threshold of 2 percentage points, indicating trends and a yet unequilibrated system. If the threshold is not met, we continue to evolve the dynamics for an additional 5000 replicator time steps and repeat the procedure until the average sharing adoption per destination becomes stationary.

Per replicator time step and per destination node we repeat the ride-sharing game for 100 times in the current configuration of $\pi(d,t)$ to generate a reliable numerical estimate for the expected utility increment of sharing $E[\Delta u(d,t)]$ that is being used to update $\pi(d,t+1)$. Consequently, we repeat the game at least two million times per parameter constellation $(S,\epsilon,\zeta,\xi)$ to obtain an estimate for the equilibrium ride-sharing adoption.

Supplementary Figure \ref{fig:FIG_Equilibration} illustrates the equilibration approach and quantifies the degree of fluctuations (due to imperfect estimation of the expected utilities) around the estimated stationary values of for the trajectories underlying the phase diagram from Fig.~4a in the Main Manuscript. Supplementary Fig. \ref{fig:FIG_Equilibration}b shows the standard deviation $\sigma_{\langle \pi^*\rangle}$ of the sharing fraction over time along the part of the trajectory used to estimate the equilibrium sharing fraction. For none of the $(S,\epsilon,\zeta,\xi)$-tuples the sharing fraction fluctuates by more than 1 percentage point around the estimated equilibrium value, hinting at a stationary equilibrium state. We observe that the vast majority of trajectories is much better equilibrated with negligible deviations from the estimated value $\langle \pi^*\rangle$ over time. Only in the low-sharing regime (mixed strategies requiring an exact balance of the incentives such that $E[\Delta u(d,t)] = 0$) and in the vicinity of the transition towards the high-sharing regime stronger fluctuations emerge, as expected.

In Supplementary Fig.~\ref{fig:FIG_Equilibration}c,d we show an exemplary trajectory for the $(S,\epsilon,\zeta,\xi)$-parameter tuple with highest degree of fluctuations around the equilibrium value, but satisfying the equilibration threshold. All other points reach their equilibrium values (up to the allowed degree of fluctuations) on shorter timescales.\\

\begin{figure}[h]
    \centering
    \includegraphics[width=1.0\linewidth]{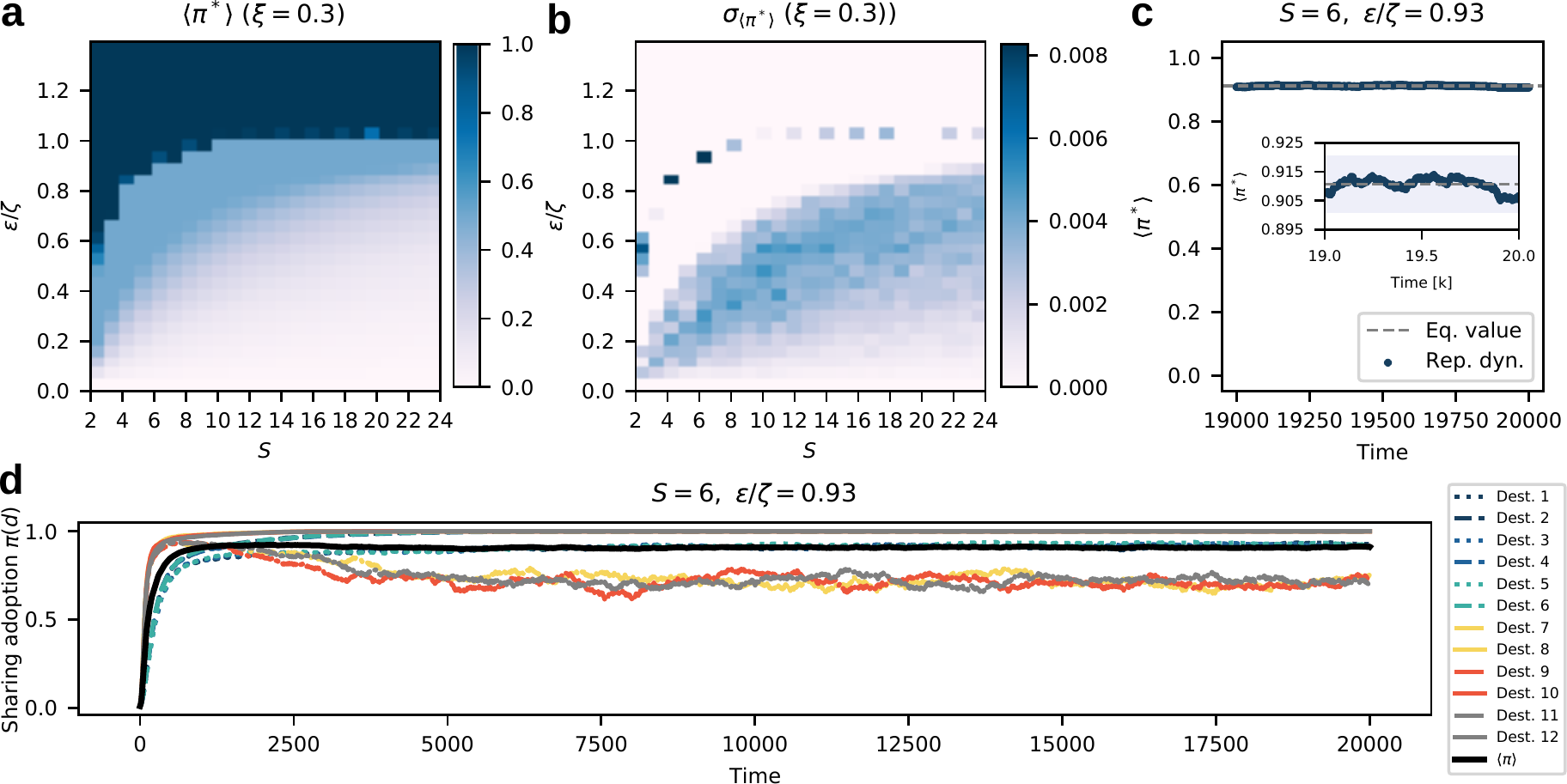}
    \caption{
    \textbf{Equilibration of replicator dynamics in the ride-sharing (anti-)coordination game.} \textbf{a} Phase diagram of equilibrium sharing fraction. \textbf{b} Temporal fluctuations in sharing fraction in the equilibrated state. \textbf{c} Time evolution of sharing fraction for the parameter constellation with longest equilibration timescale. Inset shows the part of trajectory used to determine value of stationary point (gray, dashed). Shading illustrates 1 percentage point fluctuation corridor. \textbf{d} Transient dynamics of the replicator dynamics corresponding to panel \textbf{c}.
    }
    \label{fig:FIG_Equilibration}
\end{figure}

\textit{Matching}. After generating the shared ride request graph (see Supplementary Note 3) we implement Edmond's Blossom algorithm to determine a maximum weight matching \cite{Kolmogorov2009}. Since the algorithm used implements a minimum cost perfect matching, we reduce our non-perfect matching problem to a perfect one as described in \cite[Ch. 1.5.1]{Schaefer2000}.

\end{document}